\documentclass[orivec,runningheads]{llncs}

\usepackage{iftex}

\ifpdf
  \usepackage{underscore}         
  \usepackage[T1]{fontenc}        
\else
  \usepackage{breakurl}           
\fi

\usepackage{amssymb, amsfonts, mathtools} 
\usepackage{amsmath} 
\usepackage{mathrsfs}
\usepackage{graphicx}
\usepackage{tikz}

\newcommand{\fillBox}{\hfill$\Box$}

\usetikzlibrary{shapes,arrows.meta}

\title{Causality and Decision-making: A Logical Framework for Systems and Security Modelling}

\author{Pinaki Chakraborty\inst{1}$^\ast$
\and
Tristan Caulfield \inst{1}$^\ast$
\and
David Pym\inst{1,2}$^\ast$
}
\authorrunning{P. Chakraborty, T. Caulfield, and D. Pym}

\institute{University College London, England, UK
\email{\{pinaki.chakraborty.22@ucl.ac.uk, 
t.caulfield, d.pym@ucl.ac.uk\}@ucl.ac.uk}
\and Institute of Philosophy, University of London\\ England, UK
\email{david.pym@sas.ac.uk} \\ 
$^\ast$ corresponding author
}
\titlerunning{Causality and Decision-making for Systems and Security}

\begin{document}
\maketitle

\begin{abstract}
Causal reasoning is essential for understanding decision-mak\-ing about the behaviour of complex `ecosystems' of systems that underpin modern society, with security --- including issues around correctness, safety, resilience, etc. --- typically providing critical examples. We present a theory of strategic reasoning about system modelling based on minimal structural assumptions and employing the methods of transition systems, supported by a modal logic of system states in the tradition of van Benthem, Hennessy, and Milner, and validated through equivalence theorems. Our framework introduces an intervention operator and a separating conjunction to capture actual causal relationships between component systems of the ecosystem, aligning naturally with Halpern and Pearl's counterfactual approach based on Structural Causal Models. We illustrate the applicability through examples of of decision-making about microservices in distributed systems. We discuss localized decision-making through a separating conjunction. This work unifies a formal, minimalistic notion of system behaviour with a Halpern--Pearl-compatible theory of counterfactual reasoning, providing a logical foundation for studying decision making about causality in complex interacting systems. 

\end{abstract}

\section{Introduction}
\label{sec:introduction}

Causal modelling is a multidisciplinary field spanning computer science, econometrics, epidemiology, philosophy, and statistics, providing a robust framework for understanding and reasoning about cause‐and‐effect relationships in systems. Such reasoning is of particular significance in things like root-cause analysis and strategy formulation for security. 

One influential approach to understanding causality is counterfactual analysis \cite{Lewis1973}, which stipulates that an event qualifies as a cause if, counterfactually, its absence would prevent the effect from occurring.
In many counterfactual theories of causation, directed graphs have emerged as a powerful tool for representing causal relationships, as exemplified by seminal contributions from Judea Pearl \cite{judeacausality}, Hitchcock \cite{Hitchcock2001}, and Spirtes et al. \cite{Spirtes1993}.
A formal representation of these causal relationships is provided by a set of equations known as \emph{structural equations (SE)}, which explicitly encode how each variable depends on its causal predecessors. These can be visualised as directed acyclic graphs in which vertices correspond to variables, while directed edges signify direct causal dependencies.
Building on this foundation, Halpern and Pearl \cite{HalpernPearl2005} use structural equations to define a rigorous notion of actual causation, capturing conditions under which specific events can be identified as causes of given outcomes. Although actual causality --- the objective mechanism linking causes and effects --- is distinct from our knowledge of it, which is built incrementally through observation, intervention, and counterfactual reasoning, both perspectives are deeply intertwined.
This interplay motivates our system modelling approach, in which we formalize causation through precise structural and dynamic relations aligning with Halpern and Pearl's axioms of actual causation.

In particular, our work leverages Pearl's Structural Equations approach \cite{HalpernPearl2005,Pearl2012} (a detailed discussion is deferred to \ref{sec:causalModel}) and its subsequent extension by Halpern \cite{HalpernIJCAI2015}, integrating modal logic to enable rigorous causal analysis across applications, such as mitigating risks and analysing incidents in complex systems like cyber infrastructure, where interactions among software, hardware, and human actions drive outcomes.

It should be noted that we differ from the setting of do-calculus \cite{judeacausality} in the sense that stochastic interpretation of variables, which is essential to `type causality' (see Section~\ref{subsec:actualCause}), are not involved.
Also, we introduce a uniform syntactic treatment of interventions in the logical language itself unlike the setting of Pearl who formalised interventions on a semantic level.
And lastly unlike the setting of Structural Equations in which interventions only change the value of a variable, we allow interventions to alter the structural equation relation.

It is also well established that game-based strategic reasoning about systems can be modelled using the formal technology of transition systems and, consequently, can employ the methods of process algebra --- for example, see \cite{Stirling2001} for an elegant exposition of this relationship, \cite{Sulis2024} for a reflective overview, and 
\cite{Tambe2000} for a discussion of `dynamic agent organizations', noting that `agent organizations' can be described algebraically as systems of process terms --- allowing access to the expressivity required to capture decentralized/distributed and concurrent systems. These approaches are well adapted to supporting decision-making about such systems because they naturally support a rich \emph{logical} theory that is tightly integrated with the structure of processes. Here we employ this approach in the setting of a minimalistic, behaviour-based model to discuss actual causality in `ecosystems' of interacting systems (see also \cite{Tristan2022,GLP2024}), providing tools for reasoning --- that is, decision-making --- about causation and influence between system configurations. 

We illustrate our approach with systems' security examples based on the problems involving root-cause analysis --- see, for example, \cite{NEURIPS2022RCA,Microscope2018,Groot,PHZ2024} --- that are concerned with mitigating faults arising within distributed microservice architectures in large-scale software systems (cf. \cite{AWS2023}). (See also \cite{AMP2016} for a sketch of a different approach.)

We also draw inspiration from cybernetics --- particularly Simon's work \cite{SimonAdministrativeBehavior,SimonScienceArtificial} --- which emphasizes that simple, local rules and interactions can govern complex system behaviours and dynamics.
As Simon notes in \cite{SimonAdministrativeBehavior}, `All behaviour involves conscious or unconscious selection of particular actions out of all those which are physically possible to the actor and to those persons over whom he exercises influence and authority.' This observation highlights the pivotal role of `influence' in propagating effects throughout a system. By adopting the term `influence' to describe the rules governing our system's components, we align with Simon's cybernetic tradition, viewing systems as entities shaped by local interactions.

At its core, our approach treats a system as a set of vertices --- each representing a component with observable behaviours --- whose dynamics emerges solely from a set of rules called influence.
This echoes the cybernetic insight that local interactions drive broader system dynamics, and also provides a robust platform for exploring (actual) causality in interactive environments.

Section \ref{sec:introduction} outlines the scope and necessity of causal reasoning in system modelling.
Section \ref{sec:system} introduces a minimalist approach to system modelling, followed by Section \ref{sec:logic}, which develops the logical framework used to describe system models.
In Section \ref{sec:causalModel}, we demonstrate how this logic formalizes actual causation and captures causal structures. 
Section \ref{sec:examples} explores a substantial example of how we can model decision-making about the dependencies between microservices in distributed systems. Section \ref{sec:metatheory} discusses the logical metatheory of our framework, showing how bisimulation characterizes equivalence. Finally, Section \ref{sec:discussion} situates our work within the broader landscape of causal modelling and strategic decision-making.

\section{The system modelling framework}
\label{sec:system}
In this section, we adopt a deliberately minimalist view: instead of tabulating every internal state, we specify a component only by the interactions (called its behaviour) an external observer can witness and the concrete influence those interactions have on other components. This choice is not superficial in that it draws a line between state (unobservable, intensional details) and behaviour (observable, extensional facts). 

Existing literature offers various frameworks for modelling system interactions, particularly in distributed systems --- for example, \cite{Tristan2022,AP16,CP2016,Collinson2012,Shoham2008}, with extensive relevant bibliographies, are pertinent here. 
Our work builds on this foundation by drawing inspiration from a recent abstraction \cite{GLP2024} that adopts a behaviour-centric perspective, though we incorporate dynamic aspects that extend beyond their static view. 
Our terms \textit{component, influence, configuration} are inspired from the foundational work by Winskel on event structures \cite{Winskel2009} and by Simon \cite{SimonAdministrativeBehavior}.

Formally, let $ \mathcal{C} $ be the set of components and $ \mathcal{B} $ the set of all possible behaviours.
A function $\mathbb{B}: \mathcal{C} \to 2^\mathcal{B}$ assigns to each component $ c \in \mathcal{C} $ its set of allowable behaviours, $ \mathbb{B}(c) $. 
The complete state of a system is described by a configuration that specifies each component's behaviour at a particular instant.

\begin{definition}[Configuration]
\label{def:configuration}
Let $\mathcal{C}$ be a set of components, and let $\mathbb{B}: \mathcal{C} \to 2^{\mathcal{B}}$ assign to each component $c \in \mathcal{C}$ a set $\mathbb{B}(c)$ of allowable behaviours.

A \emph{configuration} over $\mathcal{C}$ is a total function $f: \mathcal{C} \to \mathcal{B}$ such that $f(c) \in \mathbb{B}(c)$ for all $c \in \mathcal{C}$.
The set of all configurations over $\mathcal{C}$ is denoted by $F_{\mathcal{C}}$. When $\mathcal{C}$ is understood from the context, we write $F$ for brevity. \fillBox
\end{definition}

To model how components in a system evolve, we introduce influence rules that specify how component behaviours are determined --- 
formally, functions that, given the current behaviour of a component and the behaviours of selected components in the system, determine its next behaviour.

\begin{definition}[Influence rules and contexts] \label{def:influence}
Let $\mathcal{C}$ be the global set of components. For each component $c\in\mathcal{C}$, let $\mathsf{Inf}(c) \subseteq \mathcal{C}\setminus\{c\}$ be the \emph{influence context} for $c$.
$\mathsf{Inf}(c)$ is the subset of components whose behaviours are relevant for determining the behaviour of $c$.
An \emph{influence rule} for $c$ is then a function $\mathcal{I}_c: \mathbb{B}(c) \times \prod_{d\in \mathsf{Inf}(c)} \mathbb{B}(d) \to \mathbb{B}(c)$,
specifying how the behaviour of $c$ evolves from its current behaviour and the behaviours of components in its influence context.
The family of all such functions relative to a set of components $\mathcal{C}$ is called $\mathcal{I}_\mathcal{C}$. \fillBox
\end{definition}

We often omit the subscript when the referenced set of components is clear from the context. 
In our framework, a system is defined by its space of possible configurations, the transition dynamics governing their evolution, and the propositions that hold in each configuration.
This view is captured by a system model, which encapsulates the set of configurations, the transitions induced by influence rules, and the mapping of configurations to the atomic propositions that hold within them.
First, we define the transition relation which is induced by a family of influence rules.
\begin{definition}[Transition relation]
\label{def:transitionrelation}
Given a component set $\mathcal{C}$, a behaviour mapping $\mathbb{B}$, and a family of influence rules $\mathcal{I} = \{\mathcal{I}_c\}_{c \in \mathcal{C}}$ where each $\mathcal{I}_c: \mathbb{B}(c) \times \prod_{d \in \mathsf{Inf}(c)} \mathbb{B}(d) \rightarrow \mathbb{B}(c)$, the \emph{transition relation} $\Delta_{\mathcal{I}} \subseteq F \times F$ is defined as $(f, f') \in \Delta_{\mathcal{I}}$ iff there exists $c \in \mathcal{C}$ such that $f'(c) = \mathcal{I}_c(f(c),(f(d))_{d \in \mathsf{Inf}(c)})$, and, $f'(d) = f(d)$ for all $d \neq c$.
That is, $f$ transitions to $f'$ when exactly one component $c$ updates its behaviour according to its local influence rule, while all other components remain unchanged. \fillBox
\end{definition}
The definition of a system model follows:
\begin{definition}[System model]
\label{def:systemmodel}
For each $c\in \mathcal{C}$, define $\Delta_{c} = \bigl\{(f,f')\in F\times F$ such that $f'(c)= \mathcal{I}_{c}\bigl(f(c),(f(d))_{d \in \mathsf{Inf}(c)}\bigr)$.
Here, $\forall d\neq c, f'(d)=f(d)$.
Let $\Delta_{\mathcal{I}} =\bigcup_{c\in \mathcal{C}}\Delta_{c}$. A \emph{system model} $\mathcal{M}$ is a tuple $(\mathcal{C}, \mathcal{B}, \mathcal{I}, F, \Delta_{\mathcal{I}}, \Gamma)$,
where $F$ is the set of all possible configurations of the system given a set of components $\mathcal{C}$, a set of possible behaviours $\mathcal{B}$, and a family of rules $\mathcal{I}$ govern the behaviour change of the components.
$\Gamma: \mathcal{P} \rightarrow 2^{F}$ is a \textit{valuation function} that assigns a subset of the configurations to each atomic proposition from the set of atomic propositions (from a set of atomic propositions $\mathcal{P}$).

For brevity, we often omit the full notation, and write $\mathcal{M}= (F, \Delta_{\mathcal{I}}, \Gamma)$.
The transition relation $\Delta_{\mathcal{I}}$ may be viewed as the edge relation of a directed graph over the configuration space $F$, where configurations are vertices and transitions form edges.
\fillBox
\end{definition}

\begin{example}
\label{ex:configuration}
Let $\mathcal{C} = \{c_1,c_2, c_3\}$, $\mathcal{B} = \{ b_{11}, b_{12}, b_{13}, b_{21}, b_{22}, b_{31}\}$, and an assignment of behaviours be $\mathbb{B}(c_1) = \{b_{11}, b_{12}, b_{13}\}$, $\mathbb{B}(c_2) = \{b_{21}, b_{22}\}$, and $\mathbb{B}(c_3) = \{b_{31}\}$.
A configuration $f_1$ in this context is $\{(c_1,b_{11}),(c_3,b_{31})\}$.
One possible choice of the influence rules is $\mathcal{I}_{c_{1}}(b_{11}) = b_{12}$, $\mathcal{I}_{c_{1}}(b_{12}) = b_{13}$, $\mathcal{I}_{c_{1}}(b_{13}) = b_{11}$, and $\mathcal{I}_{c_{2}}(b_{21}) = b_{22}$, $\mathcal{I}_{c_{2}}(b_{22}) = b_{22}$.
Another configuration $f_2$ which is `reachable' using these rules can be $\{(c_1,b_{12}),(c_3,b_{31})\}$, and so on.
\fillBox
\end{example}

To analyse subsystems within a system model, we establish conditions under which a system can be meaningfully decomposed.
This requires identifying an interface that mediate dependencies between subsystems.
In order to define subsystems we begin with a partial configuration, which is the assignment of behaviours to only a subset of the components that constitute a full configuration.

\begin{definition}[Partial configuration]
Let $\mathcal{C}' \subseteq \mathcal{C}$ be a subset of components. A partial configuration over $\mathcal{C}'$ is a function $f': \mathcal{C}' \to \bigcup_{c \in \mathcal{C}'} \mathbb{B}(c)$, where $\mathbb{B}(c)$ is the set of possible behaviours for component $c$.
While a full configuration $f \in F$ assigns behaviours to all components in  $\mathcal{C}$, a partial configuration $f'$ assigns behaviours only to a chosen subset $\mathcal{C}'$, leaving the rest undefined.
Given a full configuration $f \in F$, its \emph{restriction} to $\mathcal{C}'$ is denoted by $f\restriction_{\mathcal{C}'}$. \fillBox
\end{definition}

The following defines an interface among the components by imposing constraints on the influence relationships among components:

\begin{definition}[Interface-admitting system model]
\label{def:interfaceadmitting}
A system model $ \mathcal{M} = (F, \Delta_{\mathcal{I}}, \Gamma) $
over a global component set $ \mathcal{C} $ is said to \emph{admit an \textbf{interface}} if there exist subsets $ C_1, C_2 \subseteq \mathcal{C} $ satisfying $C_1 \cup C_2 = \mathcal{C}$ such that for every component $ c \in C_i \setminus (C_1 \cap C_2) $ (with $ i \in \{1,2\} $), the influence context $\mathsf{Inf}(c) \subseteq C_i$, and for every component $ c \in (C_1 \cap C_2)$, the influence context $\mathsf{Inf}(c) \subseteq (C_1 \cap C_2)$.
We say that $ \mathcal{M} $ is \emph{interface-admitting} if such subsets $ C_1 $ and $ C_2 $ exist with interface $(C_1 \cap C_2)$.
The two conditions above ensure that non-interface components depend only on other components within their own partition (including the interface), and that interface components depend only on components in the interface. \fillBox
\end{definition}

\begin{remark}
For each $c \in C_i$, its local influence rule is $\mathcal{I}^i_c: \mathbb{B}(c) \times \prod_{d \in \mathsf{Inf}(c)} \mathbb{B}(d) \to \mathbb{B}(c)$.
These rules are consistent with the constraints of influence locality specified above.
The original influence rule $\mathcal{I}_c$ is recoverable from the local rules.
Specifically, for any $b \in \mathbb{B}(c)$ and any behaviour assignment $(b_d)_{d \in \mathsf{Inf}(c)}$, it holds that $\mathcal{I}_c\left(b, (b_d)_{d \in \mathsf{Inf}(c)}\right) = \mathcal{I}^i_c\left(b, (b_d)_{d \in \mathsf{Inf}(c)}\right)$. \fillBox
\end{remark}

\begin{example}[Interface]
\label{ex:interface-c2}
Let $\mathcal{C}=\{c_1,c_2,c_3\}$ and $
\mathbb{B}(c_1)=\{b_{11},b_{12},b_{13}\}$,
$\mathbb{B}(c_2)=\{b_{21},b_{22}\}$, and, $\mathbb{B}(c_3)=\{b_{31}\}$.
The influence contexts are $E(c_1)=\varnothing$,$E(c_2)=\{c_1\}$,$E(c_3)=\varnothing$. 
The influence rules are, $\mathcal{I}_{c_1}(b_{11})=b_{12}$, $\mathcal{I}_{c_1}(b_{12})=b_{13}$, $\mathcal{I}_{c_1}(b_{13})=b_{11}$, $\mathcal{I}_{c_2}(b_{21},b_{12})=b_{22}$, $\mathcal{I}_{c_2}(b_{21},\_ )=b_{21}$, $\mathcal{I}_{c_2}(b_{22},\_)=b_{22}$,$\mathcal{I}_{c_3}(b_{31})=b_{31}$. 
Thus $c_1$ cycles its behaviours autonomously, $c_2$ switches from
$b_{21}$ to the $b_{22}$ \emph{only when} $c_1$ behaves $b_{12}$, and $c_3$ is inert.
Take the partition $C_1=\{c_1,c_2\}$, $C_2=\{c_2,c_3\}$, and, $I=C_1\cap C_2=\{c_2\}$. The locality conditions of Definition~\ref{def:interfaceadmitting}
hold, and thus $\{c_2\}$ constitutes an interface. \fillBox
\end{example}

A conjugate decomposition ensures that an interface-admitting system can be partitioned into subsystems in a way that preserves the global system behaviour through consistent interactions at the interface, with local transition dynamics faithfully reflecting the overall system evolution.
\begin{definition}[Conjugate decomposition]
Let $\mathcal{M} = (F, \Delta_{\mathcal{I}}, \Gamma)$ be an inter\-face-admitting system model over a global component set $\mathcal{C}$.
In particular, let $C_1 \cap C_2$ ($C_1,C_2 \subseteq \mathcal{C}$) form an interface in $\mathcal{M}$.
A \emph{conjugate decomposition} of $\mathcal{M}$ with respect to the interface $I = C_1 \cap C_2$ is a pair of partial system models 
$(F_1, \Delta_{\mathcal{I}_1}, \Gamma_1)$ over $C_1$, and
$(F_2, \Delta_{\mathcal{I}_2}, \Gamma_2)$ over $C_2$, such that,
the following conditions are satisfied:
\begin{enumerate}
    \item $F_1$ and $F_2$ are the sets of partial configurations over $C_1$ and $C_2$, respectively, with $\Gamma_1$ and $\Gamma_2$ being the restrictions of the global valuation $\Gamma$ to $F_1$ and $F_2$.
    \item The global transition relation $\Delta_{\mathcal{I}}$ is recoverable from the partial transition relations $\Delta_{\mathcal{I}_1}$ and $\Delta_{\mathcal{I}_2}$; that is, for any full configuration $f\in F$ with restrictions $f\restriction_{C_1} = f_1$ and $f\restriction_{C_2} = f_2$, if $f \Delta_{\mathcal{I}} f'$, the corresponding restrictions satisfy $f_1 \Delta_{\mathcal{I}_1} f'_1$ and $f_2 \Delta_{\mathcal{I}_2} f'_2$, and on the interface, $I$, the partial configurations agree.
\end{enumerate} \vspace{-6mm}
\fillBox
\end{definition}

A system can intervened on by triggering changes in how its component's behaviours are altered.
This is formalized via interventions which are one-time modifications applied to a specific subset of components replacing their existing influence rules with new ones.
After the intervention, the system continues to operate under the new rules.

\begin{definition}[Intervention]
\label{def:intervention}
Consider a system model $\mathcal{M} = (F, \Delta_{\mathcal{I}}, \Gamma)$.
An \emph{intervention} $\theta_{C'}$ consists of a pair $(C', \mathcal{I}'_{C'})$, where $C' \subseteq \mathcal{C}$ is the subset of components targeted by the intervention, and $\mathcal{I}'_{C'} = \{\mathcal{I}'_c\}_{c \in C'}$ is a new set of influence rules for the components in $C'$.
Each rule $\mathcal{I}'_c : \mathbb{B}(c) \times \prod_{d \in \mathsf{Inf}(c)} \mathbb{B}(d) \to \mathbb{B}(c)$ respects the original influence context $\mathsf{Inf}(c)$.

An intervention is \emph{atomic} and \emph{one-time}: it modifies the influence rules instantaneously and irreversibly at the point of application, after which the system evolves using 
the new rules. When the intervention $\theta$ is applied, the system model transforms into $\mathcal{M}_\theta = (F, \Delta_{\mathcal{I}}^{\theta}, \Gamma)$, where the modified influence rules are given by, $\mathcal{I}^{\theta}_{c} = \mathcal{I}'_c$ if $c \in C'$; otherwise it does not change. 
The transition relation $\Delta_{\mathcal{I}}^{\theta}$ is the smallest relation closed under these revised rules, while $F$ and $\Gamma$ remain unchanged. \fillBox
\end{definition}
Note that, unlike in Structural Causal Models, where interventions fix values to some variables of interest, our framework allows interventions to replace influence rules outright, which corresponds to modifying the underlying structural equations.

\begin{example}[Intervention]
\label{ex:reset-intervention}
Let the component set be $\mathcal{C}= \{c_1,c_2,c_3\}$
and the behaviour domains be $\mathbb{B}(c_1)=\{b_{11},b_{12},b_{13}\}$, $\mathbb{B}(c_2)=\{b_{21},b_{22}\}$,and $\mathbb{B}(c_3)=\{b_{31}\}$.
The influence contexts be, $\mathsf{Inf}(c_1)=\varnothing$, $\mathsf{Inf}(c_2)=\{c_1\}$, and, $\mathsf{Inf}(c_3)=\varnothing$.
The influence rules before intervention are, $\mathcal{I}_{c_1}(b_{11})=b_{12}$,
$\mathcal{I}_{c_1}(b_{12})=b_{13}$, $\mathcal{I}_{c_1}(b_{13})=b_{11}$, $\mathcal{I}_{c_2}(b_{21},b_{12})=b_{22}$, $\mathcal{I}_{c_2}(b_{21},\_)=b_{21}$,
$\mathcal{I}_{c_2}(b_{22},\_)=b_{22}$,
$\mathcal{I}_{c_3}(b_{31})=b_{31}$.
Thus $c_1$ cycles through three states independently, $c_2$ switches from $b_{21}$ to $b_{22}$ \emph{only if} $c_1$ currently shows $b_{12}$, and $c_3$ is inert.
Apply the atomic intervention $\theta=(\{c_1\},\mathcal{I}'_{c_1})$
with $\mathcal{I}'_{c_1}(b_{11})=\mathcal{I}'_{c_1}(b_{12})=\mathcal{I}'_{c_1}(b_{13}) = b_{11}$,
leaving all other rules unchanged.
After $\theta$ every reachable configuration $f$ of the intervened model satisfies $f(c_1)=b_{11}$.
Because $c_2$ can behave $b_{22}$ only when $b_{12}$ is in its influence context,
the reset freezes $c_2$'s behaviour as $b_{21}$. \fillBox
\end{example}

\begin{remark}
If the original system $\mathcal{M}$ is decomposable via an interface-dependent decomposition, then the intervened model $\mathcal{M}_\theta$ remains decomposable under the same decomposition structure, as interventions do not alter the influence contexts, which govern the decomposition.
\fillBox
\end{remark}

In the sequel, a pointed system model is a pair consisting of a system model $\mathcal{M}$ and a \emph{chosen} configuration $f$ in the model. 

\begin{remark}
\label{remark:interveneRelation} 
Consider an intervention $\theta = (C', I'_{C'})$, where $C' \subseteq \mathcal{C}$ is the subset of components targeted by the intervention and $I'_{C'} = \{ I'_c \mid c \in C'\}$ is a set of new influence rules.
For two pointed system models $(\mathcal{M}_1,f) = (F, \Delta_{\mathcal{I}_1}, \Gamma, f)$ and  $(\mathcal{M}_2,f) = (F, \Delta_{\mathcal{I}_2}, \Gamma,f)$, we say that $\mathcal{M}_1 R_\theta \mathcal{M}_2$ holds if, for every component $c\in \mathcal{C}$,
$\mathcal{I}_2(c)= I'_c$ if $c\in C'$, and $\mathcal{I}_2(c)= \mathcal{I}_1(c)$ if  $c\notin C'$, so that the intervention changes only the influence rules for components in $C'$.
The transition relation $\Delta_{\mathcal{I}_2}$ is then induced by these updated influence rules.
We denote the union of all such relations with $R_\Theta$. \fillBox
\end{remark}

Our terms \textit{component, influence, configuration} are inspired from the foundational work by Winskel on event structures \cite{Winskel2009} and by Simon \cite{SimonAdministrativeBehavior}.
We repurpose these notions to capture dynamic causal interactions within the unified framework developed in this paper.

\section{A logic for minimal system models}
\label{sec:logic}
In this section, we introduce a logical language, denoted by $\mathcal{L}(\langle \theta \rangle, \ast)$, tailored to capture the dynamic and structural aspects of minimal system models.
Our language integrates standard modal operators $\Box$ and $\Diamond$ (with $\Diamond$ as the dual of $\Box$), a dynamic operator $\langle \theta \rangle$ that reflects interventions on a set of components, and a structural separation operator, $\ast$ --- similar in spirit to the multiplicative connective as in, for example, \cite{PymO'Hearn1999,IO2001,GLMP2025}, itself in the long tradition of relevance logic (e.g., in a vast literature, \cite{Read1988}) --- which enables the decomposition of system configurations.

\subsection{Syntax and semantics}
\label{subsec:syntax}
The language $\mathcal{L}(\langle\theta\rangle, \ast)$ is given by 
$\varphi \Coloneqq p \mid \neg \varphi \mid \varphi \land \varphi \mid \Box \varphi \mid \Diamond \varphi \mid \langle \theta \rangle \varphi \mid \varphi \ast \varphi$, 
where $p$ ranges over atomic propositions. Implication, $\rightarrow$, and disjunction, $\lor$, are defined in the usual (classical) way. 
We denote the subset consisting of $\ast$-free formulae by $\mathcal{L}(\langle\theta\rangle)$.

The semantics is defined relative to a system model $\mathcal{M} = (F, \Delta_{\mathcal{I}}, \Gamma)$ over a set of components $\mathcal{C}$.
Here, $F$ is the set of full configurations, each assigning a behaviour to every component in $\mathcal{C}$, $\Delta_{\mathcal{I}}$ is a transition relation based on influence rules 
$\mathcal{I}$, $\Gamma$ is a valuation assigning propositions to configurations.

\begin{definition}[Semantics]
\label{def:satisfaction}
Given a system model $\mathcal{M} = (F, \Delta_{\mathcal{I}}, \Gamma)$ and a configuration $f \in F$, the satisfaction relation $\models$ is defined as follows:
\[{\small 
\begin{array}{rcl}
    \mbox{$(\mathcal{M}, f) \models p$} & \;\mbox{iff}\; &  \mbox{$f \in \Gamma(p)$} \\ 
    
    \mbox{$(\mathcal{M}, f) \models \neg \varphi$} & \;\mbox{iff}\; & \mbox{$(\mathcal{M}, f) \not\models \varphi$} \\
	
    \mbox{$(\mathcal{M}, f) \models \varphi \land \psi$} & \;\mbox{iff}\; & \mbox{$(\mathcal{M}, f) \models \varphi$ and $(\mathcal{M}, f) \models \psi$} \\ 
	
    \mbox{$(\mathcal{M}, f) \models \Box \varphi$} & \;\mbox{iff}\; &  \mbox{for every $f' \in F$ with $f \Delta_{\mathcal{I}} f'$, it holds that $(\mathcal{M}, f') \models \varphi$} \\ 
	
    \mbox{$(\mathcal{M}, f) \models \Diamond \varphi$} & \;\mbox{iff}\; &  \mbox{there exists some $f' \in F$ with $f \Delta_{\mathcal{I}} f'$ such that $(\mathcal{M}, f') \models \varphi$} \\
    
    \mbox{$(\mathcal{M}, f) \models \langle \theta \rangle \varphi$} & \;\mbox{iff}\; &  \mbox{there exists some intervention $\theta_{C'}$ (with $C' \subseteq \mathcal{C}$) and a} \\ 
    & & \mbox{configuration $f'$ satisfying $f \Delta_{\mathcal{I}}^{\theta} f'$ such that $(\mathcal{M}_{\theta_{\mathcal{C}'}}, f') \models \varphi$,} \\ & & \mbox{where $\mathcal{M}_{\theta_{\mathcal{C}'}} = (F, \Delta_{\mathcal{I}}^{\theta_{\mathcal{C}'}}, \Gamma)$ is the updated model} \\
    
    \mbox{$(\mathcal{M}, f) \models \varphi \ast \psi$} & \;\mbox{iff}\; & \mbox{there exist $\mathcal{C}_1, \mathcal{C}_2 \subseteq \mathcal{C}$ such that 
    $\mathcal{C}_1 \cap \mathcal{C}_2$ constitutes}  \\ & & \mbox{an interface, and both $(\mathcal{M}_{\mathcal{C}_1}, f|_{\mathcal{C}_1}) \models \varphi$ and}   \\  & & \mbox{$(\mathcal{M}_{\mathcal{C}_2}, f|_{\mathcal{C}_2}) \models \psi$, where $\mathcal{M}_{\mathcal{C}_1}$ is the partial model over $\mathcal{C}_1$,} \\
    & & \mbox{and $\mathcal{M}_{\mathcal{C}_2}$ is the partial model over $\mathcal{C}_2$}
    \end{array}} 
   \] 
A model $\mathcal{M}$ satisfies a formula $\varphi$ at a configuration $f$ iff  $(\mathcal{M}, f) \models \varphi$. \fillBox
\end{definition}
A formula $\Box \varphi$ is read as `necessarily $\varphi$', meaning that in every configuration accessible from the current configuration via $\Delta_{\mathcal{I}}$, the formula $\varphi$ holds.
A formula $\Diamond \varphi$ is read as `possibly $\varphi$', meaning that there exists a configuration accessible from the current configuration via $\Delta_{\mathcal{I}}$ in which the formula $\varphi$ holds.
The separating conjunction $\ast$ is introduced to partition the system into overlapping subsystems, via a shared interface, enabling modular reasoning about distinct parts of the system.
A formula $\varphi \ast \psi$ is read as `$\varphi$ separating-conjoined with $\psi$', meaning that the system can be partitioned into two overlapping subsystems --- with their shared interface mediating external influences --- such that one subsystem satisfies $\varphi$ and the other satisfies $\psi$.
The intervention operator $\langle\theta\rangle$ allows us to formally represent and evaluate counterfactual modifications.
A formula $\langle\theta\rangle\varphi$ is read as `there exists an intervention $\theta$ such that after its application, the formula $\varphi$ holds in the resultant model'.

\begin{remark}
Each configuration $f$ is associated with a \textit{characteristic formula} $\chi_f = \bigwedge_{c \in \mathcal{C}} p_{c, f(c)}$, such that $(\mathcal{M}, f') \models \chi_f$ if and only if $f' = f$. This formula uniquely identifies $f$. \fillBox
\end{remark}

\section{Causal models}
\label{sec:causalModel}

Although many counterfactual frameworks for causal modelling exist --- ranging from probabilistic graphical models with soft interventions \cite{Koller2003} to process-based and mechanistic accounts --- Pearl and Halpern’s structural-equation approach offers two decisive advantages for our purposes. First, it pairs a clear graphical intuition with algebraic structural equations, allowing interventions to be represented by the simple replacement of functions; second, its formal counterfactual semantics maps cleanly to systems with rich internal dynamics. These features give us a manipulable, diagnostics-friendly framework that integrates naturally with our component--influence--configuration ontology, enabling fine-grained analysis of causation in complex, distributed systems.

\subsection{The Halpern--Pearl Framework}
\label{subsec:actualCause}
Pearl's account \cite{Pearl2012} formalizes a causal model as a tuple $M = \langle U, V, F \rangle$, where $U$ is a set of exogenous variables capturing external influences, $V$ is a set of endogenous variables representing the internal state, and $F$ is a family of functions (structural equations) of the form $v_i = f_i(pa_i, u_i)$ for $i=1,\ldots,n$, with $pa_i \subseteq V\setminus\{V_i\}$ being the minimal set of parent variables that determine $V_i$, and $u_i\subseteq U$ the corresponding exogenous inputs.
For any fixed assignment $U \coloneqq u$, these equations yield a unique solution that defines a distinct causal scenario.
Structural equations encode causal relationship by setting the left-hand variable as the effect and right-hand variables as causes, with equality signalling a directional `determined by' relationship.

Building on Pearl's approach, Halpern extends this foundation \cite{HalpernActualCausality} by focusing on an event-centric perspective, distinguishing between type causality (general patterns) and token causality (specific instances).
While our system modelling framework naturally aligns with Halpern's analytical framework on actual causality. 
Since we do not model causality using random variables, we focus on actual causality rather than type causality among configurations.

\emph{Actual causation} concerns retrospective causal claims --- asserting that an event $C$ was a cause of an effect $E$.
Halpern's extended framework distinguishes between endogenous and exogenous variables, where a causal model $ M = (S,F) $ consists of a signature $ S $ specifying variables and their possible values, and a set of structural equations $ F $ governing their interactions \cite{HalpernPearl2005}.
A causal setting is a pair $ (M, \vec{u}) $, where $ \vec{u} $ assigns values to exogenous variables, determining the behaviour of endogenous ones.
In this framework, an event $ A $ (encoded by some formula $\varphi$) is an actual cause of $ E $ (encoded by another formula $\psi$ if (i) both $A$ and $E$ occur in the actual world, (ii) in a counterfactual world where $A$ is absent but \textit{all else} remains fixed, $E$ does not occur, and (iii) $A$ is minimal, meaning no proper subset of $A$ suffices to bring about $E$.
The Halpern--Pearl (HP) definition of actual causation \cite{HalpernActualCausality} formalizes these three criteria of actual causal relationships via three clauses --- AC1, AC2, and AC3 (see the appendix for details).
In a similar manner, we introduce our notion of \emph{cause} within the context of system models, aligning our approach with the HP criteria while adapting it to the dynamics of configuration-based systems. We use a variant of {AC2($a^m$)} clause introduced in \cite{HalpernIJCAI2015}.

\begin{definition}[Cause]
\label{def:cause}
Let $\mathcal{M} = (F, \Delta_{\mathcal{I}}, \Gamma)$ be a system model, and $f_1, f_2 \in F$ be configurations over components $\mathcal{C}$.
Let $\psi_E = \bigwedge_{c \in \mathcal{C}_E} p_{c=f_2(c)}$ be the effect formula, where $\mathcal{C}_E \subseteq \mathcal{C}$ is the set of components relevant for determining the outcome.
A subset of components $ \mathcal{C}' \subseteq \mathcal{C}$, whose behaviours are fixed as in $f_1$ (i.e., $\chi_C = \bigwedge_{c\in \mathcal{C}'} p_{c=f_1(c)}$), is called a \textbf{cause} of $f_2$ from $f_1$ (denoted $Cause(f_1,f_2)$) if the following conditions hold:
    \begin{enumerate}\setlength\itemsep{-0.75mm}
    \item There exists a sequence of transitions such that $ f_1 \Delta^+_{\mathcal{I}} f_2$,
    where $\Delta^+_{\mathcal{I}}$ is the transitive closure of $\Delta_{\mathcal{I}}$, and $ f_1(c)=f_2(c) \quad \text{for all } c\in \mathcal{C}'$.
    This is expressed as $ \Diamond^+ (\psi_E \land \chi_C)$, and is an actuality condition (analogous to \textbf{AC1} in HP definitions \cite{HalpernActualCausality}).

    \item There exists a witness set $\mathcal{W} \subseteq \mathcal{C}$ such that for any configuration $f_1'$ where $f_1'(c) = f_1(c)$ for all $c \in \mathcal{W}$, but $f_1'(c) \neq f_1(c)$ for some $c \in \mathcal{C}' \setminus \mathcal{W}$, if $f_1' \Delta_{\mathcal{I}} f_2'$, then $f_2' \neq f_2$. This is the \textit{counterfactuality condition} analogous to \textbf{AC2.} in HP definitions \cite{HalpernActualCausality}).
        
    Let $\chi_C = \bigwedge_{c\in C} p_{c=f_1(c)}$, $\chi_{\mathcal{W}} = \bigwedge_{c\in \mathcal{W}} p_{c=f_1(c)}$.
    Also let 
    \[{\small 
    \chi_{C\setminus\mathcal{W}}' = \bigvee_{c\in C\setminus\mathcal{W}} \neg p_{c=f_1(c)}}
    \]
    be the formula expressing that at least one component in $C\setminus\mathcal{W}$ has changed relative to $f_1$ after an intervention.
    The counterfactual condition can be expressed as $ \langle \theta \rangle (\chi_{\mathcal{W}} \ast \chi_{C\setminus\mathcal{W}}') \to \Box^+ \neg ( \psi_E \land \chi_C )$

    \item There is no proper subset $\mathcal{C}'' \subset \mathcal{C}'$ that satisfies both the above conditions. This is the \emph{Minimality} condition analogous to \textbf{AC3} in HP  \cite{HalpernActualCausality}).
    \end{enumerate}
    \vspace{-8mm}
    \fillBox
\end{definition}

The invariance of the candidate cause's behaviours across a transition from configuration  $f_1$ to $f_2$, expressed as  $f_1(c) = f_2(c)$ for all $c \in \mathcal{C}'$, aligns with Halpern and Pearl's AC1 condition, which requires that both the candidate cause and the effect hold in the actual world.
It confirms the candidate cause's presence in the actual system evolution, enabling counterfactual analysis: by altering the candidate cause in a hypothetical scenario and observing the effect's absence, we isolate its causal role.
The second condition helps isolate the subset of components which constitute a cause.
The third conditions ensures no proper subset of the candidate cause suffices as an actual cause --- by enforcing that every component in  $\mathcal{C}'$ is essential; restricting to any proper subset $\mathcal{C}'' \subset \mathcal{C}'$ disrupts either the invariance in the actual transition or the counterfactual dependence, thus preventing over-attribution. 

Understanding how changes propagate through a system is essential for analysing causality.
A \textit{causal chain} captures this progression by linking configurations through causal dependencies (refer to Definition~\ref{def:cause}), ensuring that each transition satisfies the established criteria of causal relationships.

\begin{definition}[Causal chain]
\label{def:causalchain}
A \textit{causal chain} in a given system model $ \mathcal{M}$  $=$ $(F, \Delta_{\mathcal{I}}, \Gamma)$ is a finite sequence of configurations $(f_1, f_2, \ldots, f_n)$ with $n \geq 2$ and each $f_i \in F$, satisfying the following conditions:
\begin{enumerate} \setlength\itemsep{-0.75mm}
\item For each consecutive pair $(f_i, f_{i+1})$, there exists a subset of components $\mathcal{C}_i \subseteq \mathcal{C}$ such that $ \mathcal{C}_i$ is a cause of $ f_{i+1} $ from $ f_i $ according to the three criteria (actuality, counterfactuality, minimality). 
   
\item For each $ i$, it holds true that $(f_i, f_{i+1}) \in \Delta_{\mathcal{I}}^+ $, meaning that the causal influence is realizable through one or more transitions in the system.

\item The chain is minimal in the sense that no configuration $ f_k $ can be removed without violating sequential causality.
This ensures that the chain does not include superfluous steps.
\end{enumerate}
We denote the set of all causal chains in $ \mathcal{M} $ by $Chain(\mathcal{M})$. \fillBox
\end{definition}

\begin{definition}[Causal system model]
\label{def:CausalModel}
A causal projection of a system model $\mathcal{M} = (F,\Delta_{\mathcal{I}},\Gamma)$ is a tuple $(F^c, \Delta^c, \Gamma^c) $ such that $ F^c \subseteq F$ consists of configurations that appear in at least one causal chain in $Chain(\mathcal{M}) $, and, $\Delta^c$ and $\Gamma^c$ are the restrictions of $\Delta_{\mathcal{I}}$ and $\Gamma$ respectively to $F^c$.
The system model $\mathcal{M} = (F, \Delta_{\mathcal{I}}, \Gamma)$ is called a \emph{causal system model} if it has a causal projection.
\fillBox
\end{definition}
If the graph $(F^c, \Delta^c)$ is acyclic, i.e., $\Delta^c$ is a partial order then there are no `causal loops' (a `causal loop' is formed when for any two configurations $f_1$ and $f_2$, both $f_1$ and $f_2$ are causes of each other).

Lemma~\ref{lem:characterizationIntervention}, below, characterizes how interventions affect causal chains by delineating conditions under which such chains are either preserved or disrupted.
It is proved in the appendix (\ref{sec:app}).

\begin{lemma}[Characterization of Interventions in Causal System Models]
\label{lem:characterizationIntervention}
Let $\mathcal{M} = (F, \Delta_{\mathcal{I}}, \Gamma)$ be a causal system model with causal projection $\mathcal{M}^c = (F^c, \Delta_{\mathcal{I}}^c, \Gamma^c)$.
Let $\theta = (C', I'_{C'})$ be an intervention yielding the intervened system $\mathcal{M}_\theta = (F, \Delta_{\mathcal{I}_\theta}, \Gamma)$.
If a causal chain $ (f_1, \dots, f_n) \in \text{Chain}(\mathcal{M}) $ does not involve any component in $C'$ as part of the cause for any transition, then this chain is preserved under intervention $\theta$. Formally, $\forall i  (1 \leq i < n), C' \cap \text{Cause}(f_i, f_{i+1}) = \emptyset \Rightarrow (f_1, \dots, f_n) \in \text{Chain}(\mathcal{M}_\theta)$.
Otherwise if $ (f_1, \dots, f_n)$ contains a configuration $f_i$ whose cause involves components in $C'$, and the intervention $\theta$ modifies the influence rules such that the causal transition to $f_{i+1}$ is invalidated, then the chain is disrupted. Formally, $\exists i  (1 \leq i < n) \text{ such that } C' \cap \text{Cause}(f_i, f_{i+1}) \neq \emptyset \text{ and } \langle \theta \rangle \neg (f_i \Delta_{\mathcal{I}_\theta} f_{i+1})$. 
\fillBox
\end{lemma}

The following theorem, proved in the appendix, 
relates our approach to modelling causes in systems with the HP framework \cite{HalpernActualCausality}: 

\begin{theorem}
\label{thm:HPequivalence}
Let $\mathcal{M} = (F, \Delta_{\mathcal{I}}, \Gamma)$ be a causal system model over a finite component set $\mathcal{C}$, where causes are defined via causal chains satisfying actuality, counterfactual dependence, and minimality (cf. Definition~\ref{def:cause}). Construct a corresponding HP causal model 
$M = \langle U, V, F \rangle$,
where $V = \{ V_c \mid c \in \mathcal{C} \}$ with $\operatorname{dom}(V_c)=\mathbb{B}(c)$ and the structural equations in $F$ are induced by the influence rules $\mathcal{I}$ of $\mathcal{M}$. For any configuration $f_2 \in F$, define the effect formula 
$
\varphi_{f_2} = \bigwedge_{c \in \mathcal{C}} (V_c = f_2(c)).
$
Then, if there exists a causal chain in $\mathcal{M}$ from an initial configuration $f_1$ to $f_2$, one can extract a candidate cause; that is, a subset $\vec{X} \subseteq V$ and an assignment $\vec{x}$ (with a corresponding context $\vec{u}$) such that the assignment $\vec{X} = \vec{x}$ satisfies the HP criteria (AC1--AC3) for being an actual cause of $\varphi_{f_2}$ in $(M,\vec{u})$. In other words, the existence of a causal chain from $f_1$ to $f_2$ in $\mathcal{M}$ implies that there is a corresponding actual cause in the HP model.
\fillBox 
\end{theorem}
\begin{proof}
    Refer to the appendix (\ref{sec:app}). \fillBox
\end{proof}

\section{Security examples}
\label{sec:examples}
Modern cloud-native applications often decompose functionality into independently deployable microservices consisting of a very large number of services. Microservices architecture promotes cost optimization and sustainability by enabling selective scaling of components based on demand, minimizing resource use and waste. It also allows for smaller, independent updates, reducing the need for extensive end-to-end testing compared to monolithic systems. 
This shift from monolithic to microservice-based architectures has transformed how software is designed, deployed, and maintained \cite{Microservices2016} (cf. refer to this whiteppaer from Amazon Web Services \cite{AWS2023} for technical details).

This evolution has also intensified the need for rigorous tools in forensic analysis and audit. A framework for actual causation, grounded in Halpern's approach \cite{HalpernActualCausality} but adapted to model system transitions, is well-suited as a first step in addressing these challenges (see also \cite{Collinson2012,BCIP2024} on model design perspectives).

\subsection{Microservices}
\label{subsec:microservices} 
Since in microservice-based architecture, an application typically consists of a large number of loosely coupled, fine-grained services, accurately reconstructing inter-service call graphs is non-trivial. Dependencies evolve at runtime, and often lack static configurations. \cite{Microservices2016}.
These difficulties have motivated causal-discovery techniques in industry-facing tools \cite{NEURIPS2022RCA}, stressing the need for a rigorous framework such as the one proposed here.

Decomposition into loosely coupled services with explicit APIs (Application Program Interface) mirrors our formal notions of components, configurations, and interfaces, making microservices an ideal case study. Their ubiquity in large-scale deployments ensures industrial relevance, failures often arise from identifiable interactions among just a few services, and operational practice already employs one-shot mitigations (rolling updates, circuit breakers, traffic re-routes) that correspond to the atomic interventions in our logic.
Since failures frequently trace back to a small cluster of inter-service interactions, actual causality is a useful notion in determining the precise chain of responsibility for post-incident audits and forensic analysis in microservice deployments.

\subsection{Graph-based paradigms}
\label{subsec:graphs}
Several existing tools employ \emph{causal dependency graph} to trace how anomalies propagate through microservice ecosystems. For instance, Microscope \cite{Microscope2018} infers service dependencies in real time to build a \emph{service impact graph}, which it combines with runtime anomalies to derive a causality graph. Groot \cite{Groot}, designed for large-scale systems, constructs a global dependency graph and, upon alerts, extracts a focused subgraph around affected services. It aggregates events (e.g., CPU spikes, HTTP errors, code changes) and uses domain-specific rules to assemble an \emph{event causality graph}. While effective for diagnostics, such tools treat causality observationally and do not support formal reasoning about interventions or counterfactuals \cite{BDFJMPZ2021}.

Collectively, these systems illustrate the power of graph-based methods in \emph{heuristically} localizing root causes.
Yet, they fall short of providing a \emph{rigorous} foundation for \emph{actual causation}; that is, a precise characterization of `which component state (or event) truly \textit{caused} the observable failure', in the sense of \emph{Halpern--Pearl} counterfactual dependence \cite{HalpernPearl2005}.

We argue that, in the absence of a formal specification language for expressing actual causality in dynamic systems (in contrast to do-calculus, which is designed for causal inference), such approaches remain inadequate for purposes of audit. This limitation is particularly acute in high-stakes scenarios, such as microservice-based infrastructures in financial exchanges, where root cause analysis is often conducted through ad hoc means.
For instance, the consultation paper issued by the UK Financial Conduct Authority \cite{FCA2020} exemplifies the use of informal causal chain-based analysis for forensics and audit.

While full empirical validation is beyond the scope of this theoretical development, our framework is conceptually compatible with existing microservice monitoring tools (such as \cite{Microscope2018,Groot}), where detected anomalies correspond to particular configurations or behaviour assignments in our model. Future empirical work could involve systematically translating observed anomalies and performance metrics into formal configurations and causal chains.

\subsection{Logic-based causation models in reliability engineering}
\label{subsec:reliability}
The use of logical languages, especially modal and substructural logics, as rigorous tools to specify and reason about system properties has a long history in formal methods.
Substructural logics, such as Separation Logic \cite{IO2001,Reynolds2002}, capture notions of \emph{local reasoning}, allowing one to decompose large systems into \emph{independent} subsystems or configurations. 
By introducing a \emph{separating conjunction} ($\ast$), one can assert that two sub-configurations do not interfere. 
This precision is valuable in reasoning about microservice architectures, where container boundaries and inter-service links must be sharply distinguished. 
However, causal reasoning additionally requires a formal account of \emph{intervention}.

As discussed in Section~\ref{subsec:actualCause}, Halpern--Pearl's (HP) causal framework encodes dependencies via structural equations, modelling interventions by fixing the values of selected variables. 
In particular, the modal operator $[\mathsf{do}(X \coloneqq x)]\varphi$ expresses that, after forcibly setting variable $X$ to $x$, the proposition $\varphi$ holds \cite{HalpernIJCAI2015}.

\subsection{Modelling microservices}
\label{subsec:modelling-micro} 
In this section, we illustrate how to apply the system‐modelling framework to a small microservice deployment.
We then show how to decompose the system into subsystems with a shared interface.
We also show how strategic queries can be formulated in this approach using conterfactuals.

\subsubsection{Components and behaviours}
\label{subsec:microserviceComponents}
In a typical web application using the microservices architecture, the following design pattern is often used: $\mathsf{Auth}$ handles user authentication, $\mathsf{UserDB}$ manages credential storage and lookup, $\mathsf{ProfileSvc}$ provides user profile information, $\mathsf{Logger}$ records system events and requests, and $\mathsf{FrontEnd}$ serves as the user-facing component coordinating interactions among the back-end services. 

In our framework, this corresponds to a set of components 
\[
\mathcal{C} =\{\mathsf{Auth},\mathsf{UserDB},\mathsf{ProfileSvc},\mathsf{Logger},\mathsf{FrontEnd}\}
\]
Each component $c \in \mathcal{C}$ is associated with a set of permissible behaviours (the behaviour names are self-explanatory):
\[{\small 
\begin{array}{lcl} 
\mathbb{B}(\mathsf{Auth}) & = & \{\mathit{idle},\mathit{authSucc},\mathit{authFail}\} \\
\mathbb{B}(\mathsf{UserDB}) & = & \{\mathit{idle},\mathit{dbOK},\mathit{dbError}\}
\end{array}
\begin{array}{lcl}
\mathbb{B}(\mathsf{ProfileSvc}) & = & \{\mathit{idle},\mathit{profileOK}, \mathit{TimeOut}\} \\
\mathbb{B}(\mathsf{Logger}) & = & \{\mathit{idle},\mathit{logged},\mathit{logFail}\} \\
\mathbb{B}(\mathsf{FrontEnd}) & = & \{\mathit{idle},\mathit{serving},\mathit{error}\} 
\end{array}}
\]
A\! \emph{configuration}\! $f\!\in\!F$ is a function $f\!:\! \mathcal{C} \rightarrow \bigcup_{c\in\mathcal{C}}\mathbb{B}(c)$,
with $f(c)\!\in\! \mathbb{B}(c)$ for each $c$.

\medskip

We now specify, for each component $c\in\mathcal{C}$, an \emph{influence context} $\mathsf{Inf}(c)$ (the subset of other components whose behaviours can affect $c$), and then give a local influence rule $\mathcal{I}_c$ as in Definition \ref{def:influence}:

\begin{description}
  \item[1.] $E(\mathsf{Auth}) = \{\mathsf{FrontEnd},\mathsf{UserDB}\}$. The authentication service first receives a request from the front end; if it reaches out to the user database for credentials, then $\mathsf{UserDB}$'s state may induce a success or failure.

  \item[2.] $E(\mathsf{UserDB}) = \{\mathsf{Auth}\}$. The database processes queries only when the auth service requests it.

  \item[3.] $E(\mathsf{ProfileSvc}) = \{\mathsf{Auth},\mathsf{UserDB}\}$. The profile service fetches user data only after successful authentication and a database read.

  \item[4.] $E(\mathsf{Logger}) = \{\mathsf{Auth},\mathsf{ProfileSvc},\mathsf{FrontEnd}\}$. The logger records each request, authentication attempt, and profile lookup.

  \item[5.] $E(\mathsf{FrontEnd}) = \{\mathsf{Auth},\mathsf{ProfileSvc},\mathsf{Logger}\}$. The front-end serves pages only after successful authentication and profile data, and logs its own error or serving state.
\end{description}

Accordingly, we define local influence rules $\mathcal{I}_c:\mathbb{B}(c)\times\prod_{d\in \mathsf{Inf}(c)}\mathbb{B}(d)\to\mathbb{B}(c)$ for each $c$.
Below, we write $\mathcal{I}_c(b_c,(b_d)_{d\in \mathsf{Inf}(c)})$ for the output behaviour, given current behaviour $b_c$ of $c$ and behaviours $b_d$ of each $d\in \mathsf{Inf}(c)$.
We omit trivial cases where a component remains idle if nothing relevant changes.

\smallskip 
\noindent{\bf Authentication service} $\mathcal{I}_{\mathsf{Auth}}$ caters to all authentication activities required for interaction with external users.
  \[
  \begin{array}{rcl}
    \mathcal{I}_{\mathsf{Auth}}\bigl(\mathit{idle},(\mathit{serving},\mathit{dbOK})\bigr) & = & \mathit{authSucc} \\
    \mathcal{I}_{\mathsf{Auth}}\bigl(\mathit{idle},(\mathit{serving},\mathit{dbError})\bigr) & = & \mathit{authFail}
  \end{array}
  \]
  That is, when the front end issues a login request (modelled as $\mathsf{FrontEnd}=\mathit{serving}$) and the database is OK, then $\mathsf{Auth}$ transitions to $\mathit{authSucc}$; if the database is in $\mathit{dbError}$, then $\mathsf{Auth}$ transitions to $\mathit{authFail}$.

\smallskip 
\noindent{\bf User database} $\mathcal{I}_{\mathsf{UserDB}}$:
  \[
  \begin{array}{rcl}
    \mathcal{I}_{\mathsf{UserDB}}\bigl(\mathit{idle},(\mathit{authSucc})\bigr) & = & \mathit{dbOK} \\ 
    \mathcal{I}_{\mathsf{UserDB}}\bigl(\mathit{idle},(\mathit{authFail})\bigr) & = & \mathit{dbError}
  \end{array}
  \]
  Thus, if $\mathsf{Auth}$ has just succeeded, the database returns $\mathit{dbOK}$; if $\mathsf{Auth}$ failed, the database reports $\mathit{dbError}$.

 \smallskip 
\noindent{\bf Profile service} $\mathcal{I}_{\mathsf{ProfileSvc}}$:
  \[
    \begin{array}{rcl}
    \mathcal{I}_{\mathsf{ProfileSvc}}\bigl(\mathit{idle},(\mathit{authSucc},\mathit{dbOK})\bigr)
    & = & \mathit{profileOK} \\ 
    \mathcal{I}_{\mathsf{ProfileSvc}}\bigl(\mathit{idle},(\mathit{authFail},\_)\bigr)
    & = & \mathit{TimeOut},
  \end{array}
  \]
  where $\_$ denotes `any database state'.
  In other words, if authentication succeeds \emph{and} the database is OK, the profile lookup succeeds; if authentication fails, the profile request times out.

  \smallskip 
\noindent{\bf Logger} $\mathcal{I}_{\mathsf{Logger}}$ acts as a shared interface between other components.
  \[
    \mathcal{I}_{\mathsf{Logger}}\bigl(\mathit{idle},(b_{\mathsf{Auth}},b_{\mathsf{ProfileSvc}},\mathit{serving})\bigr) = \mathit{logged}, 
    \]
    whenever $b_{\mathsf{Auth}}\in\{\mathit{authSucc},\mathit{authFail}\},~b_{\mathsf{ProfileSvc}}\in\{\mathit{profileOK},\mathit{TimeOut}\}$ 
  \[    
      \mathcal{I}_{\mathsf{Logger}}\bigl(\mathit{idle},(b_{\mathsf{Auth}},\mathit{error})\bigr) 
        = \mathit{logFail}
  \]
  Thus, if the front end is $\mathit{serving}$ and both $\mathsf{Auth}$ and $\mathsf{ProfileSvc}$ have transitioned to some success/failure state, the logger records it ($\mathit{logged}$). If the front end itself is in $\mathit{error}$, the logger may fail to log ($\mathit{logFail}$).

  \smallskip 
\noindent{\bf Front-end} $\mathcal{I}_{\mathsf{FrontEnd}}$: $\mathcal{I}_{\mathsf{FrontEnd}}\bigl(\mathit{idle},(b_{\mathsf{Auth}},b_{\mathsf{ProfileSvc}},b_{\mathsf{Logger}})\bigr) = \mathit{serving}$, if $b_{\mathsf{Auth}}=\mathit{authSucc}$ and $b_{\mathsf{ProfileSvc}}=\mathit{profileOK}$ and $b_{\mathsf{Logger}}=\mathit{logged}$.
Otherwise, it equals $\mathit{error}$ if $b_{\mathsf{Auth}}=\mathit{authFail} \lor b_{\mathsf{ProfileSvc}}=\mathit{TimeOut} \lor b_{\mathsf{Logger}}=\mathit{logFail}$.
  In other words, the front end will serve the requested page only if authentication and profile lookup succeed and the logger has recorded those events; otherwise it enters an $\mathit{error}$ state.

\smallskip 
\noindent{\bf The system model} Collecting everything, we obtain a \emph{system model} $
\mathcal{M} = \bigl(\mathcal{C},\mathcal{B},\mathcal{I},F,\Delta_{\mathcal{I}},\Gamma\bigr)$, where 
\begin{enumerate}
  \item $\mathcal{C}$ is the component set above, 
  \item $\mathcal{B} = \bigcup_{c\in\mathcal{C}}\mathbb{B}(c)$ is the union of all behaviour sets, 
  \item $\mathcal{I} = \{\mathcal{I}_c\mid c\in\mathcal{C}\}$ is the family of influence rules just defined.
  \item $F$ is the set of all full configurations $f:\mathcal{C}\to\bigcup_c \mathbb{B}(c)$, 
  \item $\Delta_{\mathcal{I}} \subseteq F \times F$ is the one‐step transition relation induced by $\mathcal{I}$, 
  \[
      (f,f') \in \Delta_{\mathcal{I}}
      \quad\text{iff}\quad
      \forall c\in\mathcal{C},f'(c)=\mathcal{I}_c\bigl(f(c),(f(d))_{d\in \mathsf{Inf}(c)}\bigr)
  \] 
  \item $\Gamma: \mathcal{P} \to 2^{F}$ is a valuation that assigns, for each atomic proposition in a chosen propositional vocabulary $\mathcal{P}$, the set of configurations in which it holds. For example, $
      \Gamma(p_{\mathsf{FrontEnd}=\mathit{error}})
       =
      \{f\in F \mid f(\mathsf{FrontEnd})=\mathit{error}\}$,
    and similarly for propositions like $p_{\mathsf{Auth}=\mathit{authFail}}$, and so on.
\end{enumerate}

To show that $\mathcal{M}$ admits a non-trivial interface decomposition, partition $\mathcal{C}$ into $C_1 = \{\mathsf{Auth},\mathsf{UserDB},\mathsf{Logger}\}$ and $C_2 = \{\mathsf{ProfileSvc}, \mathsf{FrontEnd},\mathsf{Logger}$. Note that $C_1\cup C_2=\mathcal{C}$ and $I = C_1\cap C_2 = \{\mathsf{Logger}\}$ is the interface.
We can check the two conditions of Definition \ref{def:systemmodel} (Interface‐admitting System Model). 

\subsection{Strategic decision queries} \label{subsec:strategic-queries} 

We now show how to \emph{formalize decision‐making questions} in the microservice example without developing a full game‐theoretic apparatus. Note that it could have been formulated in the setting of a multi‐agent game.

\medskip

\noindent\textbf{Notation} Recall $\varphi_{\mathrm{fail}} = p_{\mathsf{FrontEnd}=\mathit{error}}$, and the three candidate interventions:
\[
\begin{aligned}
  \theta_1 &= \bigl(\{\mathsf{UserDB}\},\{\mathcal{I}'_{\mathsf{UserDB}}\}\bigr), 
    \quad
    \mathcal{I}'_{\mathsf{UserDB}}(\cdot) \coloneqq \mathit{dbOK} \\
  \theta_2 &= \bigl(\{\mathsf{FrontEnd}\},\{\mathcal{I}'_{\mathsf{FrontEnd}}\}\bigr), 
    \quad
    \mathcal{I}'_{\mathsf{FrontEnd}}(\_,\_) \coloneqq \mathit{servingCache} \\
  \theta_3 &= \bigl(\{\mathsf{ProfileSvc}\},\{\mathcal{I}'_{\mathsf{ProfileSvc}}\}\bigr), 
    \quad
    \mathcal{I}'_{\mathsf{ProfileSvc}}(\_,\_) \coloneqq \mathit{profileStale}
\end{aligned}
\]

\medskip

\noindent\textbf{Guaranteed recovery}  
Which interventions $\theta_i$ guarantee $\neg\varphi_{\mathrm{fail}}$ from configuration $f_2$?  Formally, $(\mathcal{M},f_2)\models\langle \theta_i \rangle \Box\neg\varphi_{\mathrm{fail}}$.
In our example, 
\[
\mbox{$(\mathcal{M},f_2)\models\langle \theta_1 \rangle \Box\neg\varphi_{\mathrm{fail}}$, 
$\mathcal{M},f_2)\models\langle \theta_2 \rangle \Box\neg\varphi_{\mathrm{fail}}$, and  
  $(\mathcal{M},f_2)\not\models\langle \theta_3 \rangle \Box\neg\varphi_{\mathrm{fail}}$}
\]
Thus $\theta_1$ (repairing the DB) and $\theta_2$ (cache‐serve) are \emph{valid recovery policies}, while $\theta_3$ is not.

\medskip
\noindent\textbf{Minimal‐cost intervention}  Suppose we assign costs to each $\theta_i$: $\mathit{Cost}(\theta_1) = 10$, $\mathit{Cost}(\theta_2) = 5$, $\mathit{Cost}(\theta_3) = 2$, where, for example, repairing the database is more expensive than re-routing to the cache. We wish to choose the $\theta_i$ that (i) satisfies $\langle \theta_i \rangle \Box\neg\varphi_{\mathrm{fail}}$ and (ii) minimizes $\mathit{Cost}(\theta_i)$. The corresponding formula might be 
\[{\small  
    (\mathcal{M},f_2)\models \langle \theta_i \rangle \Box\neg\varphi_{\mathrm{fail}}
    \;\mbox{and}\; 
    (\mathcal{M},f_2)\models \langle\zeta_j\rangle \Box\neg\varphi_{\mathrm{fail}} \;\mbox{implies}\;  (\mathit{Cost}(\theta_i) \le \mathit{Cost}(\zeta_j))
}
\]
for some $\theta_i$ and for all $\zeta_j$ (a predicate version of the logic could be used to internalize the  quantifications).
In our setting, $\theta_2$ is chosen, since $\mathit{Cost}(\theta_2)=5$ is the lowest cost among $\{\theta_1,\theta_2\}$ that guarantee recovery.

\medskip

\noindent\textbf{Fallback vs. repair trade‐off}  If $\mathit{Utility}(\theta_i)$ combines cost and the user‐satis\-faction penalty (e.g., stale data penalty), we can write  
$\mathit{Utility}(\theta_i) = -\mathit{Cost}(\theta_i) - \mathit{Penalty}(\theta_i)$,
and ask for $(\mathcal{M},f_2)\models \langle \theta_i \rangle \Box\neg\varphi_{\mathrm{fail}}$ and  $(\mathcal{M},f_2)\models \langle\zeta_j\rangle\Box\neg\varphi_{\mathrm{fail}}$
implies $(\mathit{Utility}(\theta_i)\ge \mathit{Utility}(\theta_j))$, for some $\theta_i$ and for all $\zeta_j$.
This yields the `best trade‐off' policy under a combined cost‐penalty metric.

\medskip

\noindent\textbf{Localized decision-making though separation (using $\ast$)} In this example, using the interface $\{\mathsf{Logger}\}$, we can ensure that an intervention on one subsystem does not violate invariants in the other. For instance, when applying $\theta_2$, we require 
\[
  (\mathcal{M},f_2)\models\langle \theta_2 \rangle (\varphi_{C_1} \ast \varphi_{C_2}) 
\]
where
$\varphi_{C_1} = p_{\mathsf{UserDB}=\mathit{dbError}} \land p_{\mathsf{Logger}=\mathit{logged}}$
and $\varphi_{C_2} = p_{\mathsf{FrontEnd}=\mathit{servingCache}} \land p_{\mathsf{Logger}=\mathit{logged}}$. 
This asserts that after forcing the front‐end to $\mathit{servingCache}$, the $C_1$‐subsystem (DB--Auth--Logger) can continue with $\mathsf{Logger}$ $=$ $\mathit{logged}$ and $\mathsf{UserDB}$ $=$ $\mathit{dbError}$, while the $C_2$‐subsystem (ProfileSvc--FrontEnd--Logger) enters a safe `cache' configuration.
Thus the intervention respects subsystem locality and prevents cross‐subsystem side‐effects.

\medskip
\noindent\textbf{Strategic perspective.}
This framework could have been enriched by viewing an orchestrator (defender) and external failures (attackers) as players: the defender's strategy set would be $\{\theta_{1},\theta_{2},\theta_{3},\dots\}$, while the adversary's `strategy' is the choice of which component fails next. A natural payoff function rewards the absence of failures minus the cost of interventions.
While our logical intervention--queries already suffice to guide practical decision-making without constructing a full game model, the framework naturally suggests a fuller game-theoretic treatment.
We therefore leave the explicit formulation of full strategy spaces, payoff functions, and equilibrium concepts to future work.

\subsection{Actual Causal Analysis}
\label{subsec:causal-queries}
Given two configurations, $f_1,f_2\in F$, We now exhibit how to identify a set of components $\mathcal{C}'$ as a \emph{cause} of $f_2$ from $f_1$ per Definition \ref{def:cause}.
Intuitively, $f_1$ will be read as a normal `no‐error" configuration, while $f_2$ exhibits a front‐end error. We show that a misconfiguration in $\mathsf{UserDB}$ (in $C_1$) is the actual cause of $f_2$.

\medskip

\noindent\textbf{Configuration $f_1$:} All components are idle or behave in a successful manner. All services await incoming requests.

\[
\begin{array}{lclclcl}
  f_1(\mathsf{Auth}) & = & \mathit{idle} & \quad & f_1(\mathsf{Logger}) & = & \mathit{idle} \\

  f_1(\mathsf{UserDB}) & = & \mathit{idle} & \quad &
  f_1(\mathsf{FrontEnd}) & = & \mathit{idle} \\

  f_1(\mathsf{ProfileSvc}) & = & \mathit{idle}
\end{array}
\]

\medskip

\noindent\textbf{Configuration $f_2$:}  (A front‐end error due to a database fault.)
\[
\begin{array}{lclclcl}
  f_2(\mathsf{Auth})       & = & \mathit{authFail}      & \quad &
  f_2(\mathsf{Logger})     & = & \mathit{logged}        \\[4pt]

  f_2(\mathsf{UserDB})     & = & \mathit{dbError}       & \quad &
  f_2(\mathsf{FrontEnd})   & = & \mathit{error}         \\[4pt]

  f_2(\mathsf{ProfileSvc}) & = & \mathit{profileTimeout}
\end{array}
\]

Here, the request reached the front end, $\mathsf{Auth}$ attempted to authenticate, but $\mathsf{UserDB}$ returned $\mathit{dbError}$, leading to $\mathsf{Auth}=\mathit{authFail}$, $\mathsf{ProfileSvc}=\mathit{profileTimeout}$, the logger recorded the events, and finally the front end transitioned to $\mathit{error}$.

\noindent\textbf{Effect Formula $\psi_E$.}  We consider the observable failure `FrontEnd is in error' as the effect, $\psi_E = p_{\mathsf{FrontEnd}=\mathit{error}}$.
Thus $\psi_E$ holds exactly in those configurations where $f(\mathsf{FrontEnd})=\mathit{error}$.

\noindent\textbf{Candidate Cause $\mathcal{C}'=\{\mathsf{UserDB}\}$.}  We claim that fixing $\mathsf{UserDB}$ in state $\mathit{dbError}$ (as in $f_2$) is an \emph{actual cause} of $f_2$ from $f_1$. To check Definition \ref{def:cause}, we let $\chi_C \;=\; p_{\mathsf{UserDB}=\mathit{dbError}}$
Intuitively, `$\mathsf{UserDB}$ is stuck in $\mathit{dbError}$' is our cause candidate.

\textbf{(Actuality.}  
We must show that there is a sequence of transitions from $f_1$ to $f_2$ in which $\mathsf{UserDB}$ remains $\mathit{dbError}$.
Indeed, consider the following one‐step transitions (written $f\Delta_{\mathcal{I}}f'$):

\[
\small                    
\setlength{\arraycolsep}{1pt}
\begin{array}{rcl}
  f_1 &
  \xrightarrow{\substack{\mathsf{FrontEnd}: \mathit{idle} \\ \to \mathit{serving}}} &
  f'_1\,
  \text{\parbox[t]{6.2cm}{where $f'_1(\mathsf{FrontEnd})=\mathit{serving}$, others unchanged.}}
  \\[4pt]

  f'_1 &
  \xrightarrow{\substack{\mathsf{UserDB}: \mathit{idle} \\ \to \mathit{dbError}}} &
  f'_2\,
  \text{\parbox[t]{6.2cm}{(misconfiguration injected).}}
  \\[4pt]

  f'_2 &
  \xrightarrow{\substack{\mathsf{Auth}: \mathit{idle} \\ \to \mathit{authFail}}} &
  f'_3\,
  \text{\parbox[t]{6.2cm}{since $E(\mathsf{Auth})=(\mathit{serving},\mathit{dbError})$.}}
  \\[4pt]

  f'_3 &
  \xrightarrow{\substack{\mathsf{ProfileSvc}: \mathit{idle} \\ \to \mathit{profileTimeout}}} &
  f'_4\,
  \text{\parbox[t]{6.2cm}{since $E(\mathsf{ProfileSvc})=(\mathit{authFail},\mathit{dbError})$.}}
  \\[4pt]

  f'_4 &
  \xrightarrow{\substack{\mathsf{Logger}: \mathit{idle} \\ \to \mathit{logged}}} &
  f'_5\,
  \text{\parbox[t]{6.2cm}{since it sees $(\mathit{authFail},\mathit{profileTimeout},\mathit{serving})$.}}
  \\[4pt]

  f'_5 &
  \xrightarrow{\substack{\mathsf{FrontEnd}: \mathit{serving} \\ \to \mathit{error}}} &
  f_2\,
  \text{\parbox[t]{6.2cm}{since $E(\mathsf{FrontEnd})=(\mathit{authFail},\mathit{profileTimeout},\mathit{logged})$.}}
\end{array}
\]

Throughout this run, once $\mathsf{UserDB}$ transitions to $\mathit{dbError}$, it stays in that state. Hence $\mathsf{UserDB}=\mathit{dbError}$ in all intermediate configurations, and eventually $f_2(\mathsf{FrontEnd})=\mathit{error}$. Thus $\Diamond^+(\psi_E \land \chi_C)$ holds. Moreover, in $f_1$ we indeed have $f_1(\mathsf{UserDB})=\mathit{idle}\neq\mathit{dbError}$, so the cause condition ``$\mathsf{UserDB}$ is set to $\mathit{dbError}$" is nontrivial.

\textbf{Counterfactual clause.}  
We must exhibit a witness set $\mathcal{W}\subseteq\mathcal{C}$ and show that if $\mathsf{UserDB}$ were \emph{not} set to $\mathit{dbError}$ (while keeping $\mathcal{W}$ fixed), then no run leads to $f_2$ exactly. Take $\mathcal{W} = \{\mathsf{Auth},\,\mathsf{ProfileSvc},\,\mathsf{Logger},\,\mathsf{FrontEnd}\}$. In other words, we hold all other components at their post‐failure states (in $f_2$) except $\mathsf{UserDB}$. To apply Definition \ref{def:cause}, we consider an intervention $\theta$ that forces all components in $\mathcal{W}$ to their $f_2$ values but \emph{does not} force $\mathsf{UserDB}$, allowing it to vary,
$\theta = (\mathcal{W},\{\mathcal{I}'_c : c\in\mathcal{W}\})$, $\mathcal{I}'_c\text{ simply sets }c\text{ to }f_2(c)\text{ immediately and stably.}$

Under this intervention, $\mathsf{UserDB}$ is free, and all other components behave exactly as in $f_2$.
Now, if we keep $\mathsf{Auth}=\mathit{authFail},\,\mathsf{ProfileSvc}=\mathit{profileTimeout},\,\mathsf{Logger}=\mathit{logged},\,\mathsf{FrontEnd}=\mathit{error}$ but let $\mathsf{UserDB}$ deviate from $\mathit{dbError}$ (i.e.\ $f_1'(\mathsf{UserDB})=\mathit{idle}$ or $\mathit{dbOK}$), then $\mathsf{Auth}$ could not have arrived at $\mathit{authFail}$ \emph{via} $\mathcal{I}_{\mathsf{Auth}}$ as defined, nor could $\mathsf{ProfileSvc}$ reach $\mathit{profileTimeout}$, nor could $\mathsf{FrontEnd}$ become $\mathit{error}$ under the same influence rules.
Concretely, with $\mathsf{UserDB}=\mathit{dbOK}$, one would get $\mathsf{Auth}=\mathit{authSuccess}$ and $\mathsf{ProfileSvc}=\mathit{profileOK}$, forcing $\mathsf{FrontEnd}=\mathit{serving}$. Hence no run $(f_1')\Delta_\mathcal{I} f'_2$ can yield $f'_2(\mathsf{FrontEnd})=\mathit{error}$. 
This establishes
\[
  \langle \theta \rangle \bigl(\chi_{\mathcal{W}}\ast \chi_{\{\mathsf{UserDB}\}\setminus\mathcal{W}}'\bigr) 
  \;\longrightarrow\;
  \Box^+ \neg\bigl(\psi_E \land \chi_C\bigr),
\]
verifying the counterfactual condition (AC2) of Definition \ref{def:cause}.

\textbf{Minimality.}  
Finally, no proper subset of $\{\mathsf{UserDB}\}$ is non-empty, so minimality holds vacuously.
Thus $\mathcal{C}'=\{\mathsf{UserDB}\}$ is indeed an \emph{actual cause} of $f_2$ from $f_1$.

\noindent\textbf{Interpretation.}  This formal analysis shows how a single misbehaving component in $C_1$ (the database) sufficed to produce the end‐user failure `FrontEnd error' in $C_2$, via the shared interface component `Logger.'  Because $\mathsf{Logger}$ mediates all observable events between subsystems, we can decompose the global system without losing information and still identify $\mathsf{UserDB}=\mathit{dbError}$ as the minimal actual cause of $\mathsf{FrontEnd}=\mathit{error}$.

\subsection{The necessity of the separating conjunction}
\label{subsec:why-star}

In the counterfactual clause of Definition~\ref{def:cause} we write $\langle\theta\rangle\bigl(\chi_{W} \ast \chi'_{C'\setminus W}\bigr)\;\longrightarrow\;
\Box^{+}\neg\bigl(\psi_{E}\land\chi_{C'}\bigr)$,
where $C'$ is the candidate cause set, $W\subseteq C'$ is the \emph{witness} subset, 
$\chi_{W}$ asserts that every witness component behaves exactly as in the actual run,
and $\chi'_{C'\setminus W}$ states that \emph{at least one} non-witness component now deviates.
The separating conjunction~$\ast$ is crucial: it guarantees that after an intervention~$\theta$ we can \emph{split the resulting configuration} into two overlapping sub-configurations that interact only through an explicit interface. 
Using an ordinary conjunction~$\land$ would invalidate later proofs that rely on compositionality.

In the context of our microservices example, taking the components, as before:
\[
C'=\{\mathsf{Auth},\mathsf{FrontEnd},\mathsf{Logger}\},\qquad 
W=\{\mathsf{Auth},\mathsf{Logger}\},
\]
and an intervention $\theta$ that rewrites the \textsf{FrontEnd} rule so it serves cached pages.
After~$\theta$ we obtain
\[
  \chi_{W}=p_{\mathsf{Auth}=\mathit{authSucc}}\ \land\ 
           p_{\mathsf{Logger}=\mathit{logged}},\qquad
  \chi'_{C'\setminus W}= \neg p_{\mathsf{FrontEnd}=\mathit{serving}}\ \lor\ 
                         \neg p_{\mathsf{FrontEnd}=\mathit{error}}.
\]
With~$\ast$ the post-intervention configuration decomposes into
\[
  C_{1}=\{\mathsf{Auth},\mathsf{Logger}\},\qquad
  C_{2}=\{\mathsf{FrontEnd},\mathsf{Logger}\},
\]
sharing the single interface component \textsf{Logger}. 
Locality is preserved, so the antecedent of AC2 holds. 
Replacing~$\ast$ by~$\land$ would force a \emph{single} global assignment satisfying both
$\chi_{W}$ and $\chi'_{C'\setminus W}$ simultaneously, hiding the structural separation.
This will undermine subsequent audits for root-cause analysis purpose.

The connective~$\ast$ enforces \emph{resource-sensitive locality}:  
it separates the unchanged witness region from the region where the intervention provokes change, ensuring that only the intended components vary while causal reasoning remains compositional.
In architectures such as micro-services, where subsystems are independently deployable but communicate through explicit APIs, $\ast$ Therefore is the correct logical tool for formulating the AC2 counterfactual condition.

\section{Logical metatheory}
\label{sec:metatheory}
In this section, it is established that our constructions obey a flavour of some well-established theorems that relate structural and behavioural properties of the logic.
In particular, we establish two van~Benthem-Bergstra-Hennessy-Milner \cite{Bergstra1994} theorems.
Inspired by sabotage logic \cite{Aucher2017}, we use a model-changing notion of bisimulation under intervention that extends the standard back-and-forth (zig and zag) conditions (cf. \cite{ModalLogicBook}) to account for structural modifications induced by interventions. Further details are given in 
the Appendix (\ref{sec:app}).

Theorem~\ref{thm:bisimLogicEquiv} establishes that two bisimulation equivalent interface-admitting system models are logically equivalent with respect to $\mathcal{L}(\langle\theta\rangle, \ast)$.
Theorem~\ref{thm:logicBisimEquiv} establishes the converse under specific restrictions on the language and for the subclass of interface-admitting system models with finitely many components and behaviours.
\begin{theorem}
\label{thm:bisimLogicEquiv}
If two interface-admitting pointed system models, $(\mathcal{M}_1,f_1)$ and $(\mathcal{M}_2,f_2)$ are bisimilar under intervention then for all formulae $\varphi \in L(\langle \theta \rangle,\ast)$, if $(\mathcal{M}_1,f_1) \models \varphi$ then $(\mathcal{M}_2,f_2) \models \varphi$. \fillBox
\end{theorem}
\begin{proof}
    Refer to the appendix (\ref{sec:app}). \fillBox
\end{proof}
\begin{theorem}
\label{thm:logicBisimEquiv}
Given a formula $\varphi \in \mathcal{L}(\langle\theta\rangle)$, for any two interface-admitting pointed system models $(\mathcal{M}_1,f_1)$ and $(\mathcal{M}_2,f_2)$ with finitely many components and behaviours, if $(\mathcal{M}_1,f_1) \models \varphi$ implies $(\mathcal{M}_2,f_1) \models \varphi$ then there exists a bisimulation under intervention relating $(\mathcal{M}_1,f_1)$ and $(\mathcal{M}_2,f_2)$. \fillBox
\end{theorem} 
\begin{proof}
    Refer to the appendix (\ref{sec:app}). \fillBox
\end{proof}

\section{Discussion}
\label{sec:discussion}
One key distinguishing feature of our modelling approach, compared to traditional Structural Causal Models (SCMs), is the use of interventions as explicit mechanisms to enact \emph{rule changes}, rather than merely altering variable assignments or fixed structural equations. Standard SCM approaches typically assume stable causal mechanisms represented by structural equations that remain constant throughout analysis. In contrast, our intervention modality directly modifies influence rules governing component behaviour, offering greater flexibility in modelling scenarios involving dynamic system evolution, adaptation, or deliberate operational changes.

It provides a means for decomposing and modularly reasoning about system configurations and their causal interactions, making it particularly advantageous for forensic and audit scenarios as exemplified by the microservice-based architectures example.
Beyond microservice deployments, the combination of rule-based influence modelling and substructural logic applies naturally to several high-stakes domains that demand post-incident accountability. In industrial control systems, programmable-logic controllers, sensors, and safety interlocks already expose explicit control rules; atomic overrides map conceptually to our intervention operator, while separating conjunction models isolated mixing subsystems that interact only through shared pressure signals.

In large-scale payment and trading platforms gateways, atomic update to the codebase can be treated as interventions (in some suitable localised context) letting auditors trace a settlement outage to the precise validation rule that triggered it.
Similarly in the context of blockchain-based Decentralised Finance Lending, root cause of flash-crashes (a sudden drop in the price of the underlying asset) can be traced back.

Finally, smart-grid demand-response systems feature distributed energy resources coordinated by tariff rules and load-shedding commands that also conceptually map to interventions. A common feature across these settings is that components exhibit explicit behavioural boundaries, interventions are applied as discrete, auditable actions, and attributing causal accountability is required.

However, these strengths come with certain limitations. For instance, the abstraction level chosen deliberately omits detailed timing or continuous-time dynamics, potentially restricting the granularity of analyses in cyber-physical contexts. Furthermore, empirical validation through direct mappings from real-world events or operational logs to formal configurations remains a challenge.

Some directions for future work may include, an integration with explicit probabilistic reasoning or uncertainty quantification to enhance applicability to real-world scenarios, and extension of the framework toward game-theoretic analyses, explicitly incorporating strategic decision-making and equilibrium concepts. Such extensions would significantly broaden the practical utility and theoretical robustness of our framework.

Furthermore the present work deliberately leaves the question of substructurality, how resource constraints, local perspectives and non-duplicable assumptions shape causal relations, for future study. Our longer-term goal is to refine the counterfactual semantics so that an intervention is admissible only when permitted by a subsystem's resource frame. Such a substructural extension will align naturally with our separating conjunction (capturing locality of influence) and will let us reason about responsibility and control in settings where data, authority or physical access cannot be copied, discarded or globally modified.

Finally, we remark that it would evidently be both interesting and challenging to explore the ideas presented here in the context of learning-enabled AI systems (and by extension cyber-physical systems), where questions of causality, correctness, and security are of increasing  prominence: this would be a substantial programme of research in exploring a substructural approach to causal-strategic modelling. 

\medskip 

\noindent{\bf Acknowledgements} 
Chakraborty is supported by UKRI through the Centre for Doctoral Training in Cybersecurity at UCL. Caulfield and Pym acknowledge the partial support of UKRI Research 
Grants EP/R006865/1 and EP/S013008/1. 

\nocite{*}
\bibliographystyle{splncs04}
\bibliography{paper}

\section{Appendix}
\label{sec:app}

First, the proofs of the results required to establish the operational--logical equivalence --- van~Benthem-Bergstra-Hennessy-Milner --- properties are given. Then, the results are extended to causal notions over systems (such as actual causality) and the corresponding metatheoretical framework. Halpern's three criteria for characterizing actual causal relationship are also mentioned. \par

\subsection{Equivalence}
We first present the full definition of our model-changing notion of bisimulation. Recall that a pointed system model is a pair $(\mathcal{M}, f)$, where $\mathcal{M}$ is a system model and $f$ is a configuration in $\mathcal{M}$. Moreover as mentioned in remark \ref{remark:interveneRelation}, $R_\Theta$ denotes a relation on the class of pointed system models that relates two such models when one is derived from the other through the application of an intervention operation $\theta$.

\begin{definition}
Two interface-admitting pointed system models, 
\[
\mbox{$(\mathcal{M}_1,f_1)$ $=$ $(F_1,\Delta_{\mathcal{I}_1}, \Gamma, f_1)$ and $(\mathcal{M}_2,f_2)$ $=$ $(F_2, \Delta_{\mathcal{I}_2},\Gamma_2,f_2)$,} 
\]
are bisimilar under intervention if the following conditions are satisfied: 
\begin{enumerate}
    \item \textbf{(Atom):} For any atomic proposition $p$, 
    \[
    (\mathcal{M}_1, f_1) \models p \mbox{if and only if} (\mathcal{M}_2, f_2) \models p
    \]
    \item \textbf{(Zig):} If $f_1 \Delta_{\mathcal{I}_1} f_1'$, then there exists $f_2' \in F_2$ such that $f_2 \Delta_{\mathcal{I}_2} f_2'$ and  $(\mathcal{M}_1, f_1')$ and $(\mathcal{M}_2,f_2')$ are bisimilar.
    \item \textbf{(Zag):} If $f_2 \Delta_{\mathcal{I}_2} f_2'$, then there exists $f_1' \in F_1$ such that $f_1 \Delta_{\mathcal{I}_1} f_1'$ and $(\mathcal{M}_2, f_2')$ and $(\mathcal{M}_1,f_1')$ are bisimilar.
    \item \textbf{(Zig$_\Theta$):} If $(\mathcal{M}_1, f_1) R_\Theta (\mathcal{M}_1', f_1)$ then there exists $(\mathcal{M}_2', f_2)$ such that 
    \[
    (\mathcal{M}_2, f_2) R_\Theta (\mathcal{M}_2', f_2) \;\mbox{and}\; (\mathcal{M}_1', f_1) \;\mbox{and}\; (\mathcal{M}_2', f_2)
    \]
    are bisimilar under intervention.
    \item \textbf{(Zag$_\Theta$):} If $(\mathcal{M}_2, f_2) R_\Theta (\mathcal{M}_2', f_2)$, then there exists $(\mathcal{M}_1', f_1)$ such that
    \[
    (\mathcal{M}_1, f_1) R_\Theta (\mathcal{M}_1', f_1) \;\mbox{and}\; (\mathcal{M}_2', f_2) \;\mbox{and}\; \mathcal{M}_1', f_1)
    \]
    are bisimilar under intervention.
\end{enumerate}
\vspace{-8mm}
\fillBox
\end{definition}

The two van~Benthem-Bergstra-Hennessy-Milner completeness theorems mentioned in \ref{sec:metatheory} are proved below.
We provide a proof sketch omitting tedious details.
\begin{theorem}
If two interface-admitting pointed system models, $(\mathcal{M}_1,f_1)$ and $(\mathcal{M}_2,f_2)$ are bisimilar under intervention, then for all formulae $\varphi \in L(\langle \theta \rangle,\ast)$, if $(\mathcal{M}_1,f_1) \models \varphi$, then $(\mathcal{M}_2,f_2) \models \varphi$.
\end{theorem}
\begin{proof} 
Let $\varphi \in \mathcal{L}(\langle\theta\rangle,\ast)$.
The proof is by induction on the syntax of $\varphi$.
First suppose that $\varphi$ contains no connectives.
The atom clause in definition of bisimulation under intervention covers the atomic propositions. 
Our induction hypothesis is that the implication holds for all formulae containing at most $n (n \geq 0)$ boolean connectives and modal operators.
We must now show that the implication holds for all formulae $\varphi$ containing $n + 1$ connectives and operators.
Consider the case when $\varphi$ contains no modal operators.
If $\varphi$ is of the form $\neg \psi$ then by the induction hypothesis the implication is immediate.
If $\varphi = \psi_1 \land \psi_2$, then, by the induction hypothesis, both $(\mathcal{M}_2, f_2) \models \psi_1$ and $(\mathcal{M}_2,f_2) \models \psi_2$, and thereby $(\mathcal{M}_2, f_2) \models \varphi$.

Consider the case when $\varphi = \Diamond \psi$.
Assume that $(\mathcal{M}_1, f_1) \models \Diamond \psi$.
By the semantics of $\Diamond$, there exists $f_1\Delta_{\mathcal{I}_1} f_1'$ such that $(\mathcal{M}_1,f_1') \models \psi$.
By clause $\text{\textbf{Zig}}_\Diamond$ in the definition of bisimulation under intervention, it follows that there exists $f_2'$ such that $f_2\Delta_{\mathcal{I}_2}f_2'$, and $(\mathcal{M}_1,f_1')$ and $(\mathcal{M}_2,f_2')$ are bisimilar under intervention.
By the induction hypothesis, we conclude that $(\mathcal{M}_2,f_2') \models \psi$, and consequently $(\mathcal{M}_2,f_2) \models \varphi$.
Similarly, from $(\mathcal{M}_2,f_2) \models \Diamond \psi$ we conclude $(\mathcal{M}_1,f_1) \models \Diamond \psi$ by the $\text{\textbf{Zag}}_\Diamond$ clause.
The case when $\varphi = \Box \psi$ is similar. 

Now, consider $\varphi = \psi_1 \ast \psi_2$ where $\psi_1$ and $\psi_2$ are star-free formulae.
Assume $(\mathcal{M}_1,f_1) \models \varphi$.
Then there exists a pair of pointed partial models $(\mathcal{M}_1^a,f_1^a)$ and, $(\mathcal{M}_2^b,f_1^b)$ s.t. $(\mathcal{M}_1^a,f_1^a) \models \psi_1$ and $(\mathcal{M}_1^b,f_1^b) \models \psi_2$.
By the premiss of this theorem and since the models admit of interface, we have that there exist two bisimulations under intervention such that the pairs $(\mathcal{M}_1^a,f_1^a)$, $(\mathcal{M}_2^a,f_2^a)$ and $(\mathcal{M}_1^b,f_1^b)$, $(\mathcal{M}_2^b,f_2^b)$ each are bisimilar, and the pair $\{(\mathcal{M}_2^a,f_2^a)$ , $(\mathcal{M}_2^b,f_2^b)\}$ is a conjugate decomposition of $(\mathcal{M}_2,f_2)$.
Since $\psi_1$ and $\psi_2$ are star-free, it follows that $(\mathcal{M}_2^a,f_2^a) \models \psi_1$ and $(\mathcal{M}_2^b,f_2^b) \models \psi_2$, and consequently $(\mathcal{M}_2,f_2) \models \varphi$.

Consider the case when $\varphi = \langle \theta \rangle\psi$.
Assume $(\mathcal{M}_1,f_1) \models \varphi$.
By the semantics of intervention operator there exists $(\mathcal{M}_1',f_1)$ such that $(\mathcal{M}_1',f_1) \models \psi$, and $(\mathcal{M}_1,f_1) R_\theta (\mathcal{M}_1',f_1)$ for some intervention $\theta$.
From the $\text{\textbf{Zig}}_\Theta$ clause, it follows that there exists $\mathcal{M}_2'$ such that $(\mathcal{M}_2,f_2) R_\theta (\mathcal{M}_2',f_2)$, and $(\mathcal{M}_1',f_1)$ and $(\mathcal{M}_2',f_2)$ are bisimilar under intervention.
By the induction hypothesis, we conclude that $(\mathcal{M}_2',f_2) \models \psi$, and thence $(\mathcal{M}_2,f_2) \models \varphi$. The converse follows similarly via $\text{\textbf{Zag}}_\Theta$ clause. \fillBox
\end{proof}

The other theorem establishes the converse under specific restrictions on the language and for the subclass of interface-admitting system models with finitely many components and behaviours. 
\begin{theorem}
Given a formula $\varphi \in \mathcal{L}(\langle\theta\rangle)$, for any two interface-admitting pointed system models $(\mathcal{M}_1,f_1)$ and $(\mathcal{M}_2,f_2)$ with finitely many components and behaviours, if $(\mathcal{M}_1,f_1) \models \varphi$ implies $(\mathcal{M}_2,f_1) \models \varphi$, then there exists a bisimulation under intervention relating $(\mathcal{M}_1,f_1)$ and $(\mathcal{M}_2,f_2)$.
\end{theorem}
\begin{proof}
    Let $(\mathcal{M}_1,f_1) =(F_1,\Delta_{\mathcal{I}_1}, \Gamma, f_1)$ and
    $(\mathcal{M}_2,f_2) = (F_2, \Delta_{\mathcal{I}_2},\Gamma_2,f_2)$ be the two pointed models over a set of components $\mathcal{C}$.
    Since $|\mathcal{C}|$ is finite, there are only finitely many $f_1'$ and $f_2'$ such that $f_1\Delta_{\mathcal{I}_1}f_1'$, and $f_2\Delta_{\mathcal{I}_2}f_2'$.
    Similarly, there are only finitely many interventions $\theta$ possible over the subsets of $\mathcal{C}$. Therefore given a pointed model $(\mathcal{M},f)$, there will be only finitely many $(\mathcal{M}',f)$ such that $(\mathcal{M},f) R_\Theta (\mathcal{M}',f)$ for some intervention $\theta$.
    
    For all $\varphi \in \mathcal{L}(\langle\theta\rangle)$, $(\mathcal{M}_1, f_1) \models \varphi$ iff $(\mathcal{M}_2, f_2) \models \varphi$ (by assumption), and in particular, for all atomic propositions $p$, $(\mathcal{M}_1, f_1) \models p$ iff $(\mathcal{M}_2, f_2) \models p$.

    We show the following contradiction: 
    Let $(\mathcal{M}_1,f_1) R_\Theta (\mathcal{M}_1',f_1)$. Assume that there is no $\mathcal{M}_2'$ such that for all $\varphi$, $(\mathcal{M}_1', f_1) \models \varphi$ iff $(\mathcal{M}_2', f_2) \models \varphi$.
    Let $S = \{\mathcal{N}_2 | (\mathcal{M}_2,f_2)R_\Theta(\mathcal{N}_2,f_2)\}$.
    $S$ is neither non-empty nor infinite (since there are only finitely many interventions possible).
    Thus $S = \{ \mathcal{N}_2^1, \cdots , \mathcal{N}_2^n \}$. By assumption, for every $\mathcal{N}_2^i \in S$, there exists a formula $\psi^i$ such that $(\mathcal{M}_1', f_1) \models \psi^i$ and $(\mathcal{N}_2^i, f_2) \not\models \psi^i$.
    But then there is a formula $\psi^i$ for each $\mathcal{N}_2^i$ such that $(\mathcal{M}_1,f_1) \models \bigwedge_{i \in \{1,\cdots,n\} } \langle \theta \rangle \psi^i$, but $(\mathcal{M}_2,f_2) \not\models \bigwedge_{i \in \{1,\cdots,n\} } \langle \theta \rangle \psi^i$ which is a contradiction.

    Similarly, the converse clause can be shown.
    Moreover, two similar contradictions with regards to the corresponding $\Delta$ relations, establish the standard zig and zag clauses (which follows from the fact there are only finitely many $f_1'$ and $f_2'$ such that $f_1\Delta_{\mathcal{I}_1}f_1'$ and $f_2\Delta_{\mathcal{I}_2}f_2'$).
    Thus if for all $\varphi \in \mathcal{L}(\langle\theta\rangle)$, $(\mathcal{M}_1, f_1) \models \varphi$ iff $(\mathcal{M}_2, f_2) \models \varphi$, then $(\mathcal{M}_1, f_1)$ and $(\mathcal{M}_2, f_2)$ are bisimilar under intervention. \fillBox
\end{proof}

\subsection{Actual cause}
In the sequel, we provide the definition of actual cause as found in Chapter 2 of Halpern's Actual Causality \cite{HalpernActualCausality}.
As mentioned previously, we used the \textbf{AC2($a^m$)} clause.
\begin{definition}
Given a causal setting $(M, \vec{u})$, $\vec{X} = \vec{x}$ is an \emph{actual cause} of $\varphi$ if the following three conditions hold:
\begin{itemize}
    \item[] \textbf{AC1.} $(M, \vec{u}) \models (\vec{X} = \vec{x})$ and $(M, \vec{u}) \models \varphi$.
     \item[] \textbf{AC2($a^m$).} There is a set $\vec{W}$ of variables in $V$ and a setting $\vec{x}$ of the variables in $\vec{X}$ such that if $(M, \vec{u}) \models \vec{W} = \vec{w}^*$, then $(M,\vec{u}) \models [\vec{X} \leftarrow \vec{x}, \vec{W} \leftarrow \vec{w}^*]\neg \varphi$ 
     \item[] \textbf{AC3.} $\vec{X}$ is minimal; there is no strict subset $\vec{X}'$ of $\vec{X}$ such that $\vec{X}' = \vec{x}'$ satisfies conditions AC1 and AC2, where $\vec{x}'$ is the restriction of $\vec{x}$ to the variables in $\vec{X}'$.
\end{itemize}
\vspace{-7mm}
\fillBox
\end{definition}

We now give the proof of Lemma~\ref{lem:characterizationIntervention} which relates interventions in Causal System Models:

\begin{lemma}[Characterization of Interventions in Causal System Models]
Let $\mathcal{M} = (F, \Delta_{\mathcal{I}}, \Gamma)$ be a causal system model with causal projection $\mathcal{M}^c = (F^c, \Delta_{\mathcal{I}}^c, \Gamma^c)$.
Let $\theta = (C', I'_{C'})$ be an intervention yielding the intervened system $\mathcal{M}_\theta = (F, \Delta_{\mathcal{I}_\theta}, \Gamma)$. The following holds:

\begin{enumerate}
    \item \textbf{Causal Preservation Criterion:} 
    If a causal chain $(f_1, \dots, f_n) \!\in\! \text{Chain}(\mathcal{M}) $ does not involve any component in $C'$ as part of the cause for any transition, then this chain is preserved under intervention $\theta$. Formally, we have that $\forall i  (1 \leq i < n), \quad C' \cap \text{Cause}(f_i, f_{i+1}) = \emptyset \quad \Rightarrow \quad (f_1, \dots, f_n) \in \text{Chain}(\mathcal{M}_\theta)$. 
    
    \item \textbf{Causal Disruption Criterion:} 
    If a causal chain $ (f_1, \dots, f_n) \in \text{Chain}(\mathcal{M}) $ contains a configuration $f_i$ whose cause involves components in $C'$, and the intervention $\theta$ modifies the influence rules such that the causal transition to $f_{i+1}$ is invalidated, then the chain is disrupted. Formally, we have that $\exists i  (1 \leq i < n) \text{ such that } C' \cap \text{Cause}(f_i, f_{i+1}) \neq \emptyset \text{ and } \langle \theta \rangle \neg (f_i \Delta_{\mathcal{I}_\theta} f_{i+1})$.
\end{enumerate}
\end{lemma}

\begin{proof}  
Let $ \mathcal{M} = (F, \Delta_{\mathcal{I}}, \Gamma) $ be a causal system with causal projection $ \mathcal{M}^c = (F^c, \Delta_{\mathcal{I}}^c, \Gamma^c) $, and let $ \theta = (C', I'_{C'}) $ be an intervention yielding the intervened system $ \mathcal{M}_\theta = (F, \Delta_{\mathcal{I}_\theta}, \Gamma)$.

Let $ (f_1, \dots, f_n) \in Chain(\mathcal{M})$  be an arbitrary causal chain in the original system.
Assume that none of the configurations in the chain $ (f_1, \dots, f_n) $ involve any component from $C'$ as part of their cause.
By the definition of intervention, if $C'$ is disjoint from the cause set of every transition in the chain, the influence rules governing those transitions remain unchanged.
Consequently, the transition relation $\Delta_{\mathcal{I}_\theta}$ remains identical to $\Delta_{\mathcal{I}}$ for these configurations. 
Therefore, the transitions $(f_i, f_{i+1})$ are preserved, meaning the entire chain $(f_1, \dots, f_n)$ persists in $\mathcal{M}_\theta$.
Thus, the causal chain is preserved under the intervention.

Now assume that the chain $ (f_1, \dots, f_n) $ involves a configuration $ f_i $ whose cause depends on components in $ C' $. Suppose the intervention $ \theta $ modifies the influence rules such that the cause of $ f_{i+1} $ is invalidated. 
By the definition of intervention, the updated influence rule $ I'_{C'} $ modifies how components in $ C'$ contribute to transitions. If the intervention disrupts the causal condition required for $ f_i $ to transition to $ f_{i+1} $, then the transition $ (f_i, f_{i+1}) $ no longer exists in $ \Delta_{\mathcal{I}_\theta} $.
Formally, this is captured by $ \langle \theta \rangle \neg (f_i \Delta_{\mathcal{I}_\theta} f_{i+1})$. 
Consequently, the entire causal chain $ (f_1, \dots, f_n) $ is disrupted, as the sequence of causally linked configurations is broken.
Thus, the causal disruption criterion holds. \fillBox
\end{proof}

We now prove Theorem \ref{thm:HPequivalence}.
As before, we omit the tedious details.
\begin{theorem}
Let $\mathcal{M} = (F, \Delta_{\mathcal{I}}, \Gamma)$ be a causal system model over a finite component set $\mathcal{C}$.
Construct a corresponding HP causal model $M = \langle U, V, F \rangle$, where $V = \{ V_c \mid c \in \mathcal{C} \}$ with $\operatorname{dom}(V_c)=\mathbb{B}(c)$ and the structural equations in $F$ are induced by the influence rules $\mathcal{I}$ of $\mathcal{M}$.
For any configuration $f_2 \in F$, define the effect formula $\varphi_{f_2} = \bigwedge_{c \in \mathcal{C}} (V_c = f_2(c))$.
Then, if there exists a causal chain in $\mathcal{M}$ from an initial configuration $f_1$ to $f_2$, one can extract a candidate cause, that is, a subset $\vec{X} \subseteq V$ and an assignment $\vec{x}$ (with a corresponding context $\vec{u}$) such that the assignment $\vec{X} = \vec{x}$ satisfies the HP criteria (AC1--AC3) for being an actual cause of $\varphi_{f_2}$ in $(M,\vec{u})$. In other words, the existence of a causal chain from $f_1$ to $f_2$ in $\mathcal{M}$ implies that there is a corresponding actual cause in the HP model.
\end{theorem}
\begin{proof}
Let $\mathcal{M} = (F, \Delta_{\mathcal{I}}, \Gamma)$ be a causal system model over a finite component set $\mathcal{C}$ and assume there exists a causal chain 
$(f_1, f_2, \ldots, f_n)$ in $\mathcal{M}$ with $f_1$ as the initial configuration and $f_n = f_2$ as the outcome.
By definition of a causal chain, each transition $f_i \Delta_{\mathcal{I}} f_{i+1}$ is justified by the fact that some subset $\mathcal{C}_i\subseteq \mathcal{C}$ acts as a cause for the change from $f_i$ to $f_{i+1}$, while remaining unchanged, and such that altering that subset prevents the transition (counterfactual dependence), with minimality ensured by discarding any superfluous components.

We now construct a corresponding HP model 
$M = \langle U, V, F \rangle$, where the set of endogenous variables is $V = \{V_c \mid c\in\mathcal{C}\}$, with $\operatorname{dom}(V_c) = \mathbb{B}(c)$, and the structural equations in $F$ are induced directly by the influence rules $\mathcal{I}$ (i.e., for each $c$, the equation for $V_c$ is given by a function $f_c$ reflecting $\mathcal{I}_c$ and depending on the values of the variables corresponding to the influence context $\mathsf{Inf}(c)$).
Choose $U$ so that a fixed initial assignment $\vec{u}$ yields the starting configuration $f_1$.

By the premiss of this theorem, we define the effect formula 
$\varphi_{f_2} = \bigwedge_{c\in\mathcal{C}} (V_c = f_2(c))$.
Since the causal chain in $\mathcal{M}$ guarantees that $f_2$ is reached from $f_1$ via a series of transitions that satisfy the regularity, counterfactual, and minimality conditions, we can extract a candidate cause. In particular, there exists some subset $\vec{X} \subseteq V$, corresponding to the union of the causal subsets $\mathcal{C}_i$ from each transition (or a minimal such subset), and an assignment $\vec{x}$ such that:
\begin{enumerate}
    \item In the causal setting $(M,\vec{u})$, the assignment $\vec{X} = \vec{x}$ holds, and $M$ under $\vec{u}$ satisfies $\varphi_{f_2}$.
    \item If we intervene to alter the values of $\vec{X}$ (while keeping the values for a suitable witness set fixed), then the structural equations in $F$ imply that the effect $\varphi_{f_2}$ would not obtain (i.e., for every configuration reachable under the intervention, $\varphi_{f_2}$ fails). This follows directly from the counterfactual condition in the causal chain.
    \item By the minimality condition in the causal chain, no strict subset of $\vec{X}$ would suffice to guarantee the effect under the counterfactual analysis.
\end{enumerate}

Therefore, the candidate cause $(\vec{X}=\vec{x})$ extracted from the causal chain in $\mathcal{M}$ satisfies the HP conditions (AC1--AC3) for being an actual cause of $\varphi_{f_2}$ in the HP model $(M,\vec{u})$. This completes the proof that the existence of a causal chain from $f_1$ to $f_2$ in $\mathcal{M}$ implies the existence of a corresponding actual cause in $M$. \fillBox
\end{proof}
\end{document}